\documentclass[runningheads]{llncs}

\usepackage{graphicx} % Required for inserting images

\usepackage{lineno}% for line no.
\usepackage[utf8]{inputenc}
\usepackage[T1]{fontenc}
\usepackage{todonotes}
\usepackage{amsmath,amssymb,amsfonts}
\usepackage{xcolor}
\usepackage{multirow}
\usepackage{subcaption}
\usepackage[ruled, linesnumbered]{algorithm2e}
\usepackage{float}
\let\oldnl\nl% Store \nl in \oldnl
\newcommand{\nonl}{\renewcommand{\nl}{\let\nl\oldnl}}% Remove line number for one line
\usepackage{algpseudocode}
\usepackage{comment}
\usepackage{multicol}
\usepackage{vwcol}
\usepackage{tikz}
\usepackage[most]{tcolorbox}
\usepackage{todonotes}
\usepackage{mathtools}

\usepackage{caption}
\usepackage{float}
\usepackage{wrapfig,booktabs}

\usetikzlibrary{fit}
\usetikzlibrary{patterns}
\newtheorem{exmp}{Example}

\newtheorem{rmark}{Remark}

\usepackage{enumitem}
\usepackage{comment}
\usepackage{todonotes}
\usepackage{url}

\usepackage[normalem]{ulem} %strike through text

\usepackage[autostyle, english = american]{csquotes}
\MakeOuterQuote{"}

\newcommand{\xra}[1]{\overset{#1}{\rightsquigarrow}}
\usepackage{stmaryrd} %for \llbracket

\newcommand{\sr}[1]{{\textsf{\color{red} \small{[#1 -Srinivas-]}}}}
\newcommand{\rr}[1]{{\textsf{\color{purple} \small{[#1 -Rajarshi-]}}}}

\newcommand{\ignore}[1]{{}}

\setlist{nolistsep}
\algnewcommand{\LeftComment}[1]{\Statex \(\triangleright\) #1}

\begin{document}
%\linenumbers % Enables line numbers
\title{Runtime Enforcement of Hybrid System Properties}

%\begin{comment}
\author{Mir Md Sajid Sarwar\inst{1} \and Srinivas Pinisetty\inst{1} \and Rajarshi Ray\inst{2} \and Thierry J\'eron\inst{3}}
%
%\authorrunning{F. Author et al.}
% First names are abbreviated in the running head.
% If there are more than two authors, 'et al.' is used.
%
\institute{Indian Institute of Technology Bhubaneswar, Bhubaneswar 752050, India\\
\email{sajidsarwar2011@gmail.com,spinisetty@iitbbs.ac.in} \and
Indian Association for the Cultivation of Science, Kolkata 700032, India\\
\email{rajarshi.ray@iacs.res.in} \and
Univ Rennes, Inria, CNRS, IRISA, France\\
\email{Thierry.Jeron@inria.fr}} 
%\end{comment}
%\author{}
%\institute{}

\maketitle

%\vspace{-20pt}
\begin{abstract}
Runtime enforcement has emerged as a promising approach for ensuring the safety of autonomous and cyber-physical systems operating in uncertain and dynamic environments. Unlike traditional runtime verification, runtime enforcement actively intervenes during execution to prevent property violations by modifying unsafe system behaviors. Existing enforcement frameworks primarily focus on untimed or discrete-time specifications and are often limited to delaying or suppressing events, making them inadequate for reactive systems exhibiting complex continuous dynamics. In this paper, we propose a runtime enforcement framework where safety requirements are modeled using Hybrid Automata (HA). The framework combines discrete-event editing with continuous-time monitoring to support enforcement actions such as suppression, delay, and insertion of events at arbitrary time instants. Upon observing environmental inputs, the automaton is initialized, and runtime reachability analysis is used to synthesize safe corrective actions. We formally define the enforcement problem for safety hybrid automata, establish enforceability conditions, and present an online enforcement algorithm for reactive systems. A detailed case study on an Adaptive Cruise Control (ACC) system demonstrates the effectiveness of the proposed approach in maintaining safety properties under unsafe controller behaviors. Experimental results show that the framework introduces minimal computational overhead while ensuring continuous compliance with safety requirements in real time.
\end{abstract}

\section{Introduction}
%Runtime monitoring is getting significant research interest, as formally verifying a complex autonomous system is not feasible.

The current engineering landscape for safety-critical autonomous systems is seeing a fundamental shift in how we approach safety from relying solely on design-time guarantees to incorporating run-time assurances. 
Traditional approaches, such as formal verification and simulation-based testing, attempt to establish system correctness prior to deployment.
However, complete verification is often infeasible due to the complexity of autonomous systems and the uncertainty of real-world environments. Similarly, testing depends on scenario-based simulations with simplified models, where even exhaustive success cannot guarantee real-world reliability. Moreover, systems proven correct under design assumptions may still fail when environmental conditions deviate from those assumptions. Therefore, these methods alone are insufficient and require complementary runtime monitoring and enforcement mechanisms capable of detecting and mitigating safety violations during execution.

%The current engineering landscape for safety-critical autonomous systems is seeing a fundamental shift in how we approach safety from relying solely on design-time guarantees (proving or testing before deployment that nothing will go wrong) to incorporating run-time assurances (accepting that things might go wrong, but ensuring we catch them immediately). Part of the reason for this paradigm shift is that full verification of these systems is infeasible due to their inherent complexity and the unpredictable nature of their operating environments. In contrast, testing relies on scenario-based simulations with simplified physics and sensor models. However, a 100\% success rate in simulation does not guarantee reliability in the real world.
%Even systems that are provably correct can fail if the real world deviates from the design specifications. Therefore, these methods alone are insufficient and require complementary runtime monitoring to detect errors during active deployment.

Runtime Enforcement (RE)~\cite{DBLP:journals/fmsd/FalconeMFR11,DBLP:journals/tissec/LigattiBW09,phdPinisetty15,DBLP:journals/fmsd/PinisettyFJMRN14,DBLP:journals/tissec/Schneider00} is a discipline concerned with ensuring that systems adhere to prescribed behaviors during execution. It extends the traditional scope of Runtime Verification (RV)~\cite{DBLP:journals/tosem/BauerLS11,DBLP:conf/rv/Falcone10,DBLP:journals/sttt/FalconeFM12,DBLP:series/natosec/FalconeHR13,DBLP:conf/icse/PinisettySS18} by moving beyond passive observation. While RV monitors a "black-box" system to detect property violations, RE functions as an execution modifier. By synthesizing an enforcer that actively intervenes when a system risks deviating from its specification. RE maintains system integrity through corrective actions. These interventions include halting execution to prevent violations~\cite{DBLP:journals/tissec/Schneider00}, modifying event sequences by suppressing or inserting actions~\cite{DBLP:journals/tissec/LigattiBW09}, and buffering inputs to be forwarded only when safety is guaranteed~\cite{DBLP:journals/fmsd/PinisettyFJMRN14,DBLP:journals/scl/FalconeJMP16}.

Systems under Monitor (SuM) produce execution traces or events that reflect their underlying state. The monitor evaluates these streams against specifications, commonly expressed through temporal logics like LTL~\cite{DBLP:conf/focs/Pnueli77} and STL~\cite{DBLP:conf/formats/MalerN04}, or through timed automata~\cite{DBLP:journals/scl/FalconeJMP16}. While these formalisms are advantageous for capturing the timing constraints essential to real-world applications, they possess inherent limitations. Temporal logics often provide a coarse, abstract representation that fails to capture complex, continuous system dynamics, while timed automata introduce clock-based evolution, which is suitable for a limited subclass of system behaviors, leaving more complex dynamics unaddressed.

%Hybrid Systems exhibit mixed discrete and continuous behaviors, providing high-fidelity models for real-world scenarios. %These systems are inherently complex, often requiring the modeling of non-linear dynamics via differential equations.% Commonly referred to as Cyber-Physical Systems (CPS), they integrate continuous physical processes with discrete digital control.
Hybrid Automata (HA)~\cite{ALUR19953,DBLP:conf/hybrid/AlurCHH92} are a well-known mathematical framework for the modeling and specification of hybrid systems that exhibit a combination of discrete and continuous dynamics. Hybrid automata are generalizations of timed automata, and consequently, specifications given as HA expand the complexity of specifications that one would like to effectively monitor.
%offering the expressive power to encode complex continuous behaviors that standard temporal logics cannot capture \rr{are we comparing ha with linear temporal logic? One is an automaton and the other is a specification language. How can we relate them? Do we want to compare the richness of ha with timed automaton?}. %Additionally, they are well-suited for capturing the structure of systems. 
%Consequently, enforcing specifications given as HA significantly expands the range and complexity of systems that can be effectively monitored.

RE has been extensively formalized through various automata-based models. Security Automata~\cite{DBLP:journals/tissec/Schneider00} were introduced to enforce safety properties by blocking execution sequences that violate a given specification. To broaden the scope of enforceable properties, Edit Automata~\cite{DBLP:journals/tissec/LigattiBW09} were developed, enabling the enforcer to transform input sequences by suppressing and/or inserting events. Further advancements by Falcone et al.~\cite{DBLP:journals/fmsd/FalconeMFR11} introduced mechanisms capable of buffering events and releasing them only upon the observation of a satisfying trace. However, these foundational works primarily address untimed behaviors, largely ignore the complexities of real-time systems.
While RE frameworks based on Discrete Timed Automata (DTA)~\cite{DBLP:journals/tecs/PinisettyRSATH17} and Timed Automata (TA)~\cite{DBLP:journals/fmsd/PinisettyFJMRN14,DBLP:journals/scl/FalconeJMP16} have extended enforcement to real-time properties, they possess inherent limitations. Corrective actions in DTA-based approaches~\cite{DBLP:journals/tecs/PinisettyRSATH17} are typically restricted to discrete time points, whereas dense-time TA frameworks~\cite{DBLP:journals/fmsd/PinisettyFJMRN14,DBLP:journals/scl/FalconeJMP16} are generally limited to buffering and delaying mechanisms, and they are not suitable for reactive systems. The enforcement frameworks that allow instantaneous editing of events are restricted to untimed~\cite{DBLP:journals/tissec/LigattiBW09,DBLP:conf/spin/PinisettyRSTH17}, and discrete time specifications~\cite{DBLP:journals/tecs/PinisettyRSATH17}. This leaves a critical gap in enforcement for complex systems that require proactive, continuous-time interventions to maintain safety.

%\sr{THE TA WORKS ARE NOT RESTRICTED TO DISCRETE TIME POINTS BUT THE ENFORCEMENT ACTIONS IS LIMITED TO BUFFER/DELAY..}
%\red{Add a paragraph discussing existing RE/monitoring works for hybrid/CPS. There the policies to be monitored/enforced are expressed using basic automata/ DTA/ etc. Motivate why considering the policies to be enforced as HA will be advantaguous.}

\begin{figure}[htbp]
    \vspace{-20pt}
    \captionsetup{font=scriptsize}
    \centering
    \includegraphics[width=0.5\textwidth]{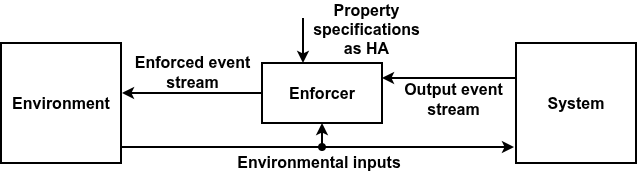}
    \caption{Runtime Enforcement Framework.}
    \label{fig:enforcer}
    \vspace{-20pt}
\end{figure}

% \begin{wrapfigure}{r}{0.5\textwidth}
%   \begin{center}
%     \includegraphics[width=0.48\textwidth]{images/enforcer.drawio.png}
%   \end{center}
%   \caption{Runtime Enforcement Framework.}
%   \label{fig:enforcer}
% \end{wrapfigure}

%\red{Figure.. SYSTEM $->$ ENF :output event stream; ENF $->$ ENV: Corrected/enforced output stream }
Utilizing hybrid automata for monitor specifications is particularly advantageous for real-time systems, as it enables dense-time monitoring between discrete sampling points. This capability is critical for maintaining safety across continuous-time trajectories.
%Furthermore, to enhance the framework's versatility, we adopt a parametric modeling approach, allowing the monitor to be generalized and instantiated across various operational contexts.
In this work, we present a \emph{runtime enforcement framework} wherein system specifications are formalized as \emph{hybrid automata} (HA). This formalism explicitly captures continuous flows via differential equations, discrete execution steps via transition guards, and operational boundaries using state invariants. Upon receiving an environmental observation, the HA is initialized, and an enforcement strategy is then dynamically synthesized by performing reachability analysis on the automaton.
%\rr{very confusing to me. RA is done on which automaton? set of safety properties?}
As illustrated in Figure~\ref{fig:enforcer}, the enforcer functions as an additional layer of security around the system that intercepts and regulates interactions between the system and its environment.

%In this work, we present a \emph{runtime enforcement framework} wherein system specifications are defined as \emph{parametric hybrid automata} (PHA)~\cite{DBLP:conf/itsc/ChaiWLLH19}.
%This formalism incorporates parametric expressions for continuous flows, discrete transition guards, and state invariants. Upon receiving an observation, the PHA is instantiated into a standard hybrid automaton.
%The inputs to the enforcer are these property specifications, inputs from the environment, and a stream of events/commands from the system or the controller of the system, and the enforcer outputs a stream of events/commands. 
The enforcer for a given property $\varphi$ takes as input a stream of inputs from the environment, and a stream of events/commands from the system, and it produces a stream of events/commands as output.
It ensures proactively that the output sequence of events adheres to the specification properties and thus ensures a safe execution in the environment.
%Most of the earlier works in this direction where enforcers are designed to suppress or delay~\cite{DBLP:journals/fmsd/PinisettyFJMRN14,DBLP:journals/scl/FalconeJMP16} an event so that the specification properties are not violated. In this work, we also consider inserting/editing an event whenever necessary to ensure the property is not violated.
%{The existing works in this direction supporting dense timed properties are restricted to delaying and suppressing events as possible enforcement actions and thus are not suitable for reactive systems~\cite{DBLP:journals/fmsd/PinisettyFJMRN14,DBLP:journals/scl/FalconeJMP16}. The enforcement frameworks that allow instantaneous editing of events are restricted to untimed~\cite{DBLP:journals/tissec/LigattiBW09,DBLP:conf/spin/PinisettyRSTH17}, and discrete time specifications~\cite{DBLP:journals/tecs/PinisettyRSATH17}.
With respect to enforcement actions, it allows instantaneous editing at discrete steps/observation points, and in addition, insertion of events between two consecutive observation points whenever necessary to ensure the policies are not violated.
%In this work, we consider policies to be expressed as PHA, and with respect to enforcement actions, allow instantaneous editing at discrete steps/observation points, and in addition, insertion of events between two consecutive observation points whenever necessary to ensure the policies are not violated.\rr{We earlier said that PHA are used as specifications. Here, we are saying that PHA are used to express policy of enforcement. Confusing to me!}
In summary, the key contributions of this work are:
\begin{itemize}
    \item We propose a novel \emph{runtime enforcement framework} for reactive CPS based on HA specifications, enabling continuous-time monitoring and corrective intervention during system execution.
    \item We formally define the \emph{runtime enforcement problem} for safety properties specified using safety hybrid automata, extending existing runtime enforcement approaches beyond untimed and discrete-time formalisms.
    \item We present an online enforcement algorithm capable of synthesizing safe execution traces through \emph{event suppression}, \emph{delay}, and \emph{insertion} while respecting system timing constraints and maintaining reactive behavior.
\end{itemize}

\section{Motivating Example} \label{sec:motive}
{In this section, we illustrate the intended behavior of an enforcer synthesized from a property specified as a hybrid automaton using the proposed approach via a simplified scenario related to a reactive adaptive cruise control (ACC) system of an autonomous car.}
%In this section, we present a motivating example for runtime enforcement (RE) policy synthesis focused on a \textcolor{blue}{reactive} adaptive cruise control (ACC) system of an autonomous car. 
%The primary safety specification requires the vehicle to maintain a safe distance to avoid a collision with the lead vehicle. We consider a scenario where the vehicle’s primary controller is faulty, potentially issuing commands that jeopardize safety. The role of the enforcer is to intercept these events and apply corrective measures—such as suppressing, delaying, or inserting new commands—to prevent a collision preemptively. \textcolor{blue}{In the reactive setup, where every clock tick, an event is observed from the incoming event stream and instantaneously rectified (suppression/delaying also handled with this), when the event violates the monitoring property. In the time period between ticks, we also allow events to be inserted, e.g., when invariants are violated.}
The primary safety specification dictates that the vehicle maintain a safe distance to prevent collisions. We assume a scenario where the vehicle’s primary controller is potentially faulty, issuing commands that jeopardize system safety. The enforcer’s role is to intercept these events and apply corrective operators—specifically suppression, delay, or insertion—to maintain safety invariants. {In our reactive configuration, the enforcer monitors the input stream at discrete clock ticks, where observed events are synchronously rectified through suppression or delay if they violate the specification. To bridge the gap between these discrete observations, the framework also facilitates inter-tick interventions. This allows the enforcer to proactively insert new commands into the stream whenever a continuous-time invariant violation is predicted during the interval between two sampling points.}
%\sr{\color{magenta}{We can see to rephrase the above to explain that we consider the reactive setup where every tick events are observed and instantaneously rectified (suppression/delaying also handled with this). In the time period between ticks, we also allow events to be inserted e.g., when...}}

%The car has 3 variables: $a$ and $v$ represent the acceleration and speed of the car, whereas $d$ represents the distance with respect to the front car. Let $-3\leq a \leq 3$, where a negative value of $a$ indicates the car is decelerating. To be on the safe side, the car should maintain a distance $d>10$ from the front car.
%The controller can generate events from the set of actions, $\Sigma$ = $\{acc$, $dec$, $cru$, $br$, $st$, $\textcolor{blue}{dummy}\}$. $acc$ and $dec$ will accelerate and decelerate the car, respectively. $cru$ will move the car at constant speed, $br$ is for making brake, whereas $st$ will stop the car. \textcolor{blue}{$dummy$ represents a dummy action which does nothing but allow continuation of the previous event.}
%Let $v_f$ denote the velocity of the preceding car. For the purpose of this analysis, we assume $v_f$ remains at a constant speed of 50 km/h, which is approximately 13.89 m/s.

The vehicle state is defined by three variables: $a$ and $v$, representing acceleration and velocity respectively, and $d$, representing the relative distance to the preceding vehicle. The acceleration is constrained to the interval $[-3, 3]\, \text{m/s}^2$, where negative values denote deceleration. 
%To satisfy the safety property $\varphi$, the vehicle must maintain a minimum distance of $d > 10$ meters from the leading car.
{Let us consider the safety property to be monitored $\varphi$ to be: ``the vehicle must maintain a minimum distance of $d > 10$ meters from the leading car''.}

The controller operates over an alphabet of discrete actions $\Sigma$ = $\{acc$, $dec$, $cru$, $br$, $st$, $\textit{stutter}\}$. Specifically, $acc$ and $dec$ initiate acceleration and deceleration, $cru$ maintains a constant velocity, and $br$ and $st$ execute braking and stopping sequences, respectively. The $\textit{stutter}$ action represents a stuttering transition, which performs no state change but facilitates the logical continuation of the preceding event. For this analysis, the velocity of the preceding vehicle ($v_f$) is assumed to be a constant 50 km/h ($\approx 13.89$ m/s).

\begin{figure}[htbp]
    \captionsetup{font=scriptsize}
    \vspace{-20pt}
    \centering
    \begin{subfigure}[b]{0.44\textwidth}
        \centering
        \includegraphics[width=0.75\textwidth]{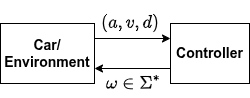}
        \caption{\scriptsize{The car without an Enforcer.}}
        \label{fig:without-enforcer}
    \end{subfigure}
    \hfill
    \begin{subfigure}[b]{0.55\textwidth}
        \centering
        \includegraphics[width=0.95\textwidth]{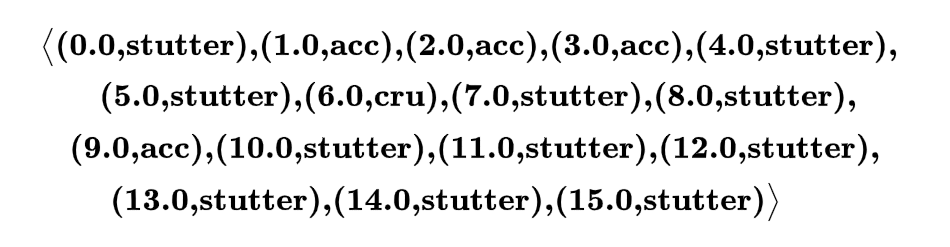}
        \caption{\scriptsize{The events stream from the controller.}}
        \label{fig:event-stream}
    \end{subfigure}
    \caption{The ACC system without a runtime enforcement layer and an example event stream from the controller that could lead to a safety violation.}
    \label{fig:unsafe-example}
    \vspace{-20pt}
\end{figure}

\begin{exmp}
%Consider a scenario when there is no enforcer, as in Figure~\ref{fig:without-enforcer}, and the output of the controller is directly feed into the environment.
%Initially, assume that the speed $v$ of the car is $10\ \text{m/s}$, and the distance $d$ from the front car is $95$ meters.
%\textcolor{blue}{The ACC to be a reactive system, which means it acts in synchronization with the outside environment. In each clock tick, the controller senses the environmental inputs and outputs its command.}
%Let assume, the controller has generated an events stream (shown in Figure~\ref{fig:event-stream}, \textcolor{blue}{the events not shown in clock ticks are $dummy$ events}). For each event in the stream, the first component is the time instance, and the second component is an action/command for the car to follow. For example, the event $(1.0, acc)$ specifies that the action $acc$ should be applied on time point $1.0$ to accelerate the car. Figure~\ref{fig:unsafe-run} shows that the application of these events during execution leads the car to unsafe conditions where the distance from the front car reduces to less than 10 meters.
Consider the scenario depicted in Figure~\ref{fig:without-enforcer}, where the controller's output is fed directly into the environment without an intervening enforcement layer. Assume that initially the car's speed is $v = 10\ \text{m/s}$ and the distance from the preceding vehicle is $d = 95$ meters.

We consider the ACC as a reactive system that operates synchronously with its environment; at each discrete clock tick, the controller samples environmental inputs and emits a control command. Suppose the controller generates the event stream illustrated in Figure~\ref{fig:event-stream}, where time intervals between explicit actions are populated by $\textit{stutter}$ events. These $\textit{stutter}$ actions function as stuttering transitions, maintaining the system's current state until the next command. Each event is defined as a tuple $(\tau, \alpha)$, representing a time instance $\tau$ and an action $\alpha$. For example, the tuple $(1.0, acc)$ mandates that the acceleration action be initiated at $t=1.0$. As demonstrated in Figure~\ref{fig:unsafe-run}, the unmediated execution of this event stream results in a safety violation, as the distance $d$ eventually falls below the $10$-meter threshold.
\end{exmp}

\begin{figure}[htbp]
    \captionsetup{font=scriptsize}
    \vspace{-20pt}
    \centering
    \includegraphics[width=0.85\textwidth]{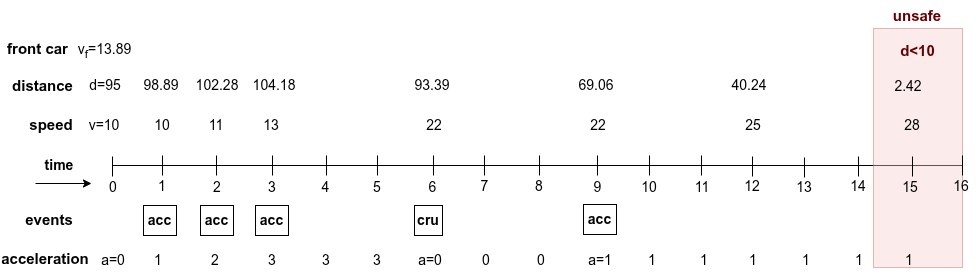}
    \caption{Execution of the input events stream shown in Figure~\ref{fig:event-stream} without an enforcer. \emph{distance}, \emph{speed}, and \emph{acceleration} show the value of the variables $d$, $v$, and $a$, respectively, at different time points throughout the execution. $v_f$ represents the constant speed of the front car. The \emph{events} below the time horizon show the application point of the actions, which leads to a safety violation. {N.B., \emph{stutter} events are not shown in figure.}}
    \label{fig:unsafe-run}
    \vspace{-20pt}
\end{figure}

%To avoid the above-described unsafe scenario, we want to design a runtime enforcer that will act like an added security layer between the car and the ACC system, as shown in Figure~\ref{fig:car-enforcer}. For the monitoring property "\textit{The distance from the front car should be greater than 10}", we design a hybrid automaton $\mathcal{H}$ (see Figure~\ref{fig:monitor-ha}) that has 5 locations along with an unsafe trap location which is not shown in the figure. From each of these locations in the figure, the control can make a transition to the unsafe location if $d<10$ (the distance from the front car).
To prevent the unsafe scenario described above, we introduce a runtime enforcer that functions as an additional safety layer between the vehicle and the ACC system (see Figure~\ref{fig:car-enforcer}). To monitor the safety requirement—“the distance to the front vehicle must remain greater than 10 meters”—we construct a safety hybrid automaton $\mathcal{S}$ (shown in Figure~\ref{fig:monitor-ha}) based on the observed environmental inputs. The automaton consists of five operational locations, along with an implicit unsafe trap location (not shown in the figure). From any of the depicted locations, the system transitions to this unsafe state whenever the condition $d < 10$ is violated, where $d$ denotes the distance to the front vehicle.
This design enables continuous monitoring of the safety constraint and immediate detection of unsafe conditions during execution.

\begin{figure}[htbp]
    \captionsetup{font=scriptsize}
    \vspace{-20pt}
    \centering
    \includegraphics[width=0.5\textwidth]{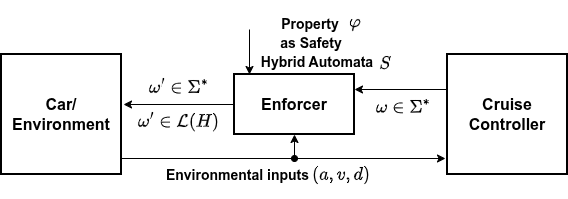}
    \caption{The enforcer for the ACC system takes in an output event stream from the controller, and produces a modified/enforced output event stream for the environment that always satisfies the property $\varphi$.}
    \label{fig:car-enforcer}
    \vspace{-20pt}
\end{figure}
%\sr{Input to enforcer is $w\in \Sigma^*$, and output is $w'\in L(H)$?}

%\noindent 
%The enforcer observes the current speed $v$ and the distance $d$ of the car. For a received events stream from the controller, it modifies (by suppressing/delaying/inserting) events to satisfy the monitoring property.
By observing the vehicle's speed $v$, acceleration $a$, and following distance $d$, the enforcer evaluates the controller’s event stream in real-time. It then transforms this stream by suppressing, delaying, or inserting events as needed—to guarantee that the system's behavior satisfies safe execution in $\mathcal{S}$.
%\sr{TO discuss reg how to present possible enforcement actions: suppressing, delaying, or inserting events... }

In \cite{DBLP:journals/tecs/PinisettyRSATH17}, a runtime enforcement framework based on discrete timed automata (DTA)~\cite{DBLP:conf/charme/BozgaMT99} is proposed. There, property specifications are expressed as DTA, and enforcement is achieved by observing events at discrete time points (ticks) and instantaneously editing them when necessary to produce a corresponding event in the same tick. In contrast, by modeling properties as hybrid automata, our framework enables continuous monitoring between clock ticks. Consequently, it can delay or insert events at arbitrary time instants, providing finer-grained control and ensuring property satisfaction in real time.

%\todo[inline]{Compare with DTA work which only suppresses or delays an event, and what additionally we bring.}

\begin{figure}[htbp]
    \captionsetup{font=scriptsize}
    \vspace{-20pt}
    \centering
    \includegraphics[width=0.70\textwidth]{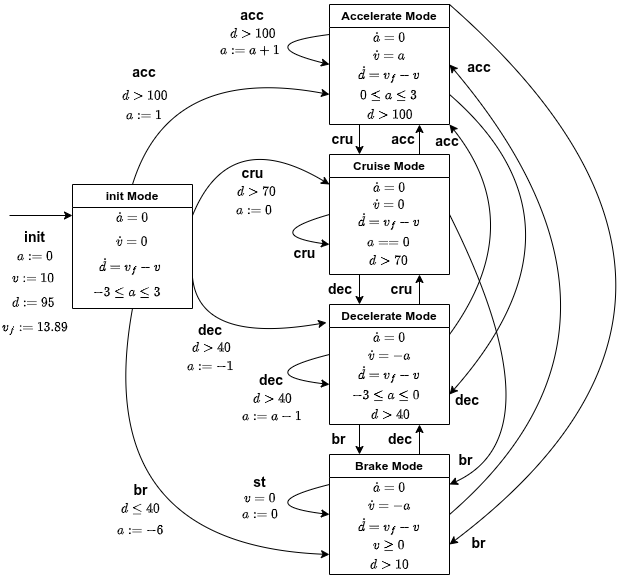}
    \caption{The monitoring property $\varphi$ as Safety Hybrid Automaton $\mathcal{S}$. The non-accepting trap location is not shown in the figure due to space constraints.}
    \label{fig:monitor-ha}
    \vspace{-10pt}
\end{figure}

\begin{exmp}
Consider the input event stream illustrated in Figure~\ref{fig:event-stream} and the corresponding enforced sequence in Figure~\ref{fig:modified-events}. Initially, the enforcer resides in the $initMode$ location (Figure~\ref{fig:monitor-ha}). Given the initial environment state $(d=95, v=10, a=0)$, the first controller event $(1.0, acc)$ cannot be executed in $\mathcal{S}$ because the distance is below the required acceleration threshold. Consequently, the enforcer {emits a stutter event at $t=1$} and \textbf{delays} this action until $t=1.29$,
%\rr{why? there are other valid choices of t as well. Can enforcer choose any?}
triggering a transition to $AccelerateMode$.The subsequent events, $(2.0, acc)$ and $(3.0, acc)$, are accepted without modification as the distance remains $d > 100$, satisfying the invariant for $AccelerateMode$. However, at approximately $t=5.22$, the distance drops below $100$ meters, violating the current location's invariant. Since no controller event is available at this time, the enforcer \textbf{inserts} $(5.22, cru)$ to transition the system into $CruiseMode$.
%\rr{What if the PHA had no outgoing transition?}.
At $t=6.0$, the controller’s $(6.0, cru)$ event is permitted as the system is already in $CruiseMode$ with $d=95.64$. Later, at $t=9.0$, the enforcer \textbf{suppresses} the received $(9.0, acc)$ event {(emitting a stutter event instead)}, as acceleration is disallowed at $d=79.17$. Finally, at $t \approx 10.67$, the enforcer \textbf{inserts} $(10.67, dec)$ to mitigate an impending invariant violation.
%\rr{How does the enforcer know at that time to insert the dec action?}. 
This synthesized event sequence ensures the safety property $d > 10$ is strictly maintained throughout the execution. {N.B., \emph{stutter} events are not shown in the figure due to space constraints.}
\end{exmp}

\begin{figure}[htbp]
    \captionsetup{font=scriptsize}
    \vspace{-20pt}
    \centering
    \includegraphics[width=0.85\textwidth]{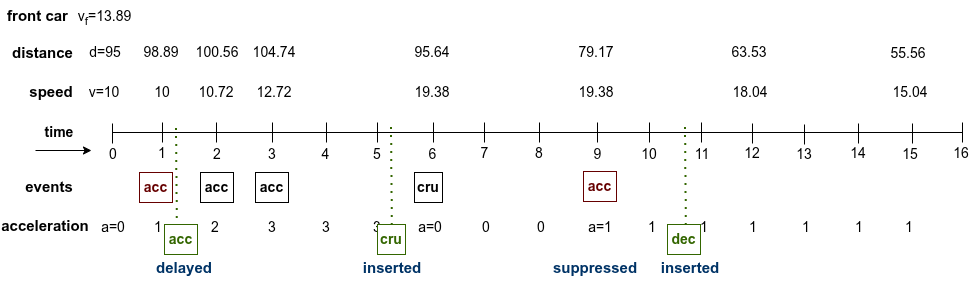}
    \caption{Output of the enforcer. It modified the event stream by suppressing/delaying/inserting events to satisfy the property $\varphi$. N.B. $stutter$ events are not shown in the figure due to space constraints.}
    \label{fig:modified-events}
    \vspace{-20pt}
\end{figure}

% \begin{figure}[htbp]
%     %\captionsetup{font=scriptsize}
%     \centering
%     \includegraphics[width=0.85\textwidth]{images/modified-run-case.drawio.png}
%     \caption{Output of the enforcer. It modified the events stream by suppressing/delaying/inserting events so that the monitoring property is satisfied.}
%     \label{fig:modified-events}
% \end{figure}
\section{Preliminaries and Notations}

%\sr{Should this part be in the beginning of prelim section?}
Let $\Sigma$ be a finite alphabet of actions. A finite word over $\Sigma$ is a finite sequence $\sigma = a_1 a_2 \ldots a_n$ where each $a_i \in \Sigma$. The sets of all finite words over $\Sigma$ are denoted by $\Sigma^*$.
%Let $\Sigma$ be a finite alphabet of actions. A finite (respectively, infinite) word over $\Sigma$ is a finite sequence $\sigma = a_1 a_2 \ldots a_n$ (respectively, an infinite sequence $\sigma = a_1 a_2 \ldots$) where each $a_i \in \Sigma$. The sets of all finite and infinite words over $\Sigma$ are denoted by $\Sigma^*$ and $\Sigma^\omega$, respectively.
The length of a finite word $\sigma$ is denoted by $|\sigma|$. The empty word is denoted by $\epsilon$.
For words $\sigma$ and $\sigma'$, their concatenation is denoted by $\sigma \cdot \sigma'$ where "$\cdot$" represents the usual concatenation of words. A word $\sigma'$ is a prefix of $\sigma$, written $\sigma' \preceq \sigma$, if there exists a word $\sigma'' \in \Sigma^*$ such that $\sigma = \sigma' \cdot \sigma''$. In this case, $\sigma$ is said to extend $\sigma'$. 

\noindent We now discuss the background concepts of hybrid systems, which exhibit an interplay of discrete and continuous dynamics. Hybrid automata are a well-known mathematical model for such systems \cite{ALUR19953,DBLP:conf/hybrid/AlurCHH92}. We first define a hybrid automaton model of a hybrid system.
%\textcolor{black}{We need to partition the state spaces into locations, or "symbolic states", which represent a hybrid system's control modes. A location, combined with the values of its continuous variables, fully defines the hybrid system's complete state.
%We first define a hybrid automaton model that facilitates the partitioning of the state space of a hybrid system.
%We begin by defining a hybrid automaton model that helps us partition the hybrid system's state space.}

\begin{definition} [Hybrid Automata] \label{def:HA} \cite{ALUR19953}
A hybrid automaton (HA) is a seven-tuple HA=\big(\emph{Loc}, \emph{Var}, \emph{Flow}, \emph{Init}, \emph{$\Sigma$}, \emph{Edge}, \emph{Inv}\big) where:
\begin{itemize}
    \item \emph{Loc} is a finite set of vertices called locations.
    \item \emph{Var} is a finite set of real-valued variables. A valuation $v$ is an assignment of a real value to each variable $x \in$ \emph{Var}. We write $V$ for the set of all valuations.
    \item 
    %\textcolor{black}{\emph{Flow} is a mapping from each location $l \in \emph{Loc}$ to a set of differential equations $\{ \dot{x} \in f(l), \forall x \in \emph{Var} \}$, where $\dot{x}$ denotes the rate of change of variable $x$.}
    \emph{Flow} is a mapping from each location $l \in \emph{Loc}$ to a set of differential equations $\{\dot{x} \in f(x_1,\ldots,x_{|Var|}) \mid x \in \emph{Var}\}$, where $\dot{x}$ denotes the rate of change of variable $x$.
    \item \emph{Init} is a tuple $\langle l_{0}, S \rangle$ such that $l_{0} \in \emph{Loc}$ and $S \subseteq V$.
    \item \emph{$\Sigma$} is a finite set of labels or actions {(including a stutter action $d$)}.
    \item \emph{Edge} is a finite set of transitions $e = (l, a, g, r, l')$, each consisting of a source location $l \in \emph{Loc}$, a target location $l'\in \emph{Loc}$, a label $a\in$ \emph{$\Sigma$}, a guard $g \subseteq V$ and a reset map $r: \mathbb{R}^{|Var|} \to 2 ^ {\mathbb{R}^{|Var|}}$.
    \item \emph{Inv} is a mapping from each location $l\in \emph{Loc}$ to a subset of valuations $V$. $\Box$
\end{itemize}
\end{definition}

%\rr{Use terms from the literature. These are called \emph{stutter} transitions. Please check whether stutter transitions have no reset, or whether they have identity reset maps. no guard or it is $\mathbb{R}^{|Var|}$. The definition uses set-theoretic representations for guards and invariants, and functions for reset maps. We need to be consistent in describing the \emph{stutter} transitions.}
\begin{rmark}
    Stutter actions are modeled as stutter transitions. A stutter transition is defined by the tuple $e = (l, d, g, r, l)$, where the guard $g = \mathbb{R}^{|Var|}$ and the reset $r$ is an identity mapping. A stutter transition for every location is implicitly a member of \emph{Edge}.
\end{rmark}
%\rr{The acceptance condition (see Muller acceptance condition - Page 6 of [3]) followed by the definition of a language of a hybrid automaton must be mentioned, with citation.}

%\noindent \textbf{Parametric hybrid automata (PHA)~\cite{DBLP:conf/itsc/ChaiWLLH19} extend standard hybrid automata by introducing parametric notations. Within the PHA framework, flow equations, transition guards, and invariants can be defined using parametric expressions, allowing the model to represent variables with initially unknown values. For instance, in the thermostat model~\cite{DBLP:conf/lics/Henzinger96} as shown in Figure~\ref{fig:thermostat}, the initial condition might be defined as $x = k_1$ and a transition guard as $x \geq 2k_2$, where $k_1$ and $k_2$ serve as parameters without fixed values in the base model. A parameter is a special variable and is distinct from system variables. The flow rate of a parameter is 0 at all locations in the automaton.

Consider the thermostat model~\cite{DBLP:conf/lics/Henzinger96} illustrated in Figure~\ref{fig:thermostat}, which serves as a foundational example of a hybrid automaton. The model features two discrete locations (or modes of operation) within which the temperature variables $x$ evolve continuously. The system executes a discrete transition between these locations whenever the corresponding guard condition is satisfied.

% \begin{figure}[htbp]
%     %\captionsetup{font=scriptsize}
%     \centering
%     \includegraphics[width=0.5\textwidth]{images/thermostat1.drawio.png}
%     \caption{A thermostat automaton.}
%     \label{fig:thermostat}
% \end{figure}

A state of a HA is defined as a pair $(l, v)$, where $l \in Loc$ represents a discrete location and $v \in \textit{Inv}(l)$ denotes a valuation of the continuous variables. The component $\textit{Init}$ specifies the set of initial configurations.
A hybrid automaton of dimension $n$ (number of variables) is an infinite-state machine whose state has a discrete part, which ranges over the locations in $Loc$, and a continuous part (valuation of $v$), which ranges over the $n$-dimensional Euclidean space $\mathbb{R}^n$~\cite{DBLP:conf/hybrid/AlurCHH92}.
A state can change either due to the application of a transition or due to the passage of time following the continuous dynamics. The continuous state change is given by $flow_{l}(v, t)$, which is the solution of the differential equation Flow($l$). We refer to the former as \emph{discrete-transition} and the latter as \emph{timed-transition}. A behaviour of the hybrid automaton is defined by a \emph{run}, which is a sequence of these discrete-transition and timed-transition steps. We define run as follows:
%\rr{Do we need this definition in this paper?}
%The discrete dynamics of a hybrid automaton can be isolated by extracting its underlying transition graph, effectively abstracting away the continuous evolution associated with each state. We define the abstract graph as follows:

% \begin{definition} [Graph]
%     \textcolor{black}{The graph of a HA is defined as $\mathcal{G_{H}} = (V$, $E)$, where $V$ = \emph{Loc} and $E \subseteq \emph{Loc} \times \Sigma \times \emph{Loc}$ such that for every $(l, a, g, r, l') \in \emph{Edge}$, there is an edge $(l\xrightarrow{a} l') \in E$ where $\{l,l'\}\in \emph{Loc}$ and $a \in \Sigma$. $\Box$}
% \end{definition}

%\noindent A hybrid automaton of dimension $n$ (number of variables) is an infinite-state machine whose state has a discrete part, which ranges over the vertices of the graph $\mathcal{G}$, and a continuous part (valuation of $v$), which ranges over the $n$-dimensional Euclidean space $\mathbb{R}^n$~\cite{DBLP:conf/hybrid/AlurCHH92}.
%A state can change either due to the application of a transition or due to the passage of time following the continuous dynamics. The continuous state change is given by $flow_{l}(v, t)$, which is the solution of the differential equation Flow($l$). We refer to the former as \emph{discrete-transition} and the latter as \emph{timed-transition}. A behaviour of the hybrid automaton is defined by a \emph{run}, which is a sequence of these discrete-transition and timed-transition steps. We define run as follows:

\begin{definition}[Run] \label{def:run}
A \emph{run} of a HA is an alternating sequence of timed-transitions and discrete-transitions of the hybrid automaton depicted as: $(l_{0},v_{0}) \xra{\tau_0}(l_{0},v'_0) \xrightarrow{a_1} (l_1,v_{1}) \xra{\tau_1} (l_1,v'_{1})\xrightarrow{a_2} \ldots \xrightarrow{a_{n}}(l_{n},v_n) \xra{\tau_{n}} (l_{n},v'_{n})$,
% \begin{align*}
% (l_{0},v_{0}) &\xra{\tau_0}(l_{0},v'_0) \xrightarrow{a_1} (l_1,v_{1}) \xra{\tau_1} (l_1,v'_{1})\xrightarrow{a_2} \ldots \xrightarrow{a_{n}}(l_{n},v_n) \xra{\tau_{n}} (l_{n},v'_{n})
% \end{align*}
where
(i) $(l_{0},v_{0}) \in \emph{Init}$.
(ii) Transitions labeled by $a_i$ represent discrete transitions, where $a_i$ denotes the label (or action) of an edge $e_i \in \emph{Edge}$. For each $i \in \llbracket 1..n \rrbracket$, the location $l_{i-1}$ is the source and $l_i$ is its destination of $e_i$. Moreover, the valuation $v'_{i-1}$ satisfies the guard $g$, and the valuation $v_i$ results from applying the reset map $r$ associated with $e_i$, i.e., $v'_{i-1} \in g$ and $v_i \in r$.
(iii) Transitions labeled by $\tau_i \in \mathbb{R}^{\geq 0}$ correspond to timed transitions, where $\tau_i$ denotes the dwelling time spent in location $l_i$. They satisfy the invariant condition that, for all $t \in [0,\tau_i]$, the continuous evolution $(l_i, v_i) \xra{t} (l_i, v_i^+)$ remains within the invariant of the location, i.e., $v_i^+ \in \mathrm{Inv}(l_i)$, for all $i \in \llbracket 0..n \rrbracket$.
(iv) $\langle t_i, a_i \rangle$ is a pair where $t_i = \sum_{j=0}^{i-1} {\tau_j}$, $\forall i \in \llbracket 1..n\rrbracket$. $t_i$ is being the time-stamp at which $a_i$ is executed.
%The transitions labeled $\tau_i \in \textcolor{black}{\mathbb{R}^{\geq 0}}$ represent timed transitions with $\tau_i$ being the time of dwelling in the location, with the constraint that $\forall t \in [0, \tau_i]$, the timed transition $(l_i,v_{i}) \xrightarrow{t} (l_i,v^+_{i})$ has $v^+_{i} \in Inv(l_i)$, $\forall i \in \llbracket 0..n\rrbracket$. %The constraint says that a state resulting from any timed transition for a duration less than or equal to $\tau_i$ must satisfy the invariant of the location.
$\Box$
\end{definition}

% \begin{definition} [Projection of Run on Graph]
% A projection of a run $r$ on the graph $\mathcal{G}_H$ of a hybrid automaton HA is a timed word $\sigma =$ $(t_1,a_1).(t_2,a_2)\cdots (tn,a_n)$ , where for each discrete-transition $(l_{i-1},v_{i-1})\xrightarrow{a_i} (l_i,v_{i})$ $\in r$, there is an edge $(l\xrightarrow{a} l') \in E$ such that $\{l,l'\}\in \emph{Loc}$ and {$a_i \in \Sigma$}, where $t_i = \sum_{j=0}^{i-1} {\tau_j}$, $\forall i \in \llbracket 1..n\rrbracket$. $\Box$
% \end{definition}

% \begin{definition} [Projection of a Run]
%    Let $r = (l_{0},v_{0}) \xra{\tau_0}(l_{0},v'_0) \xrightarrow{a_1} (l_1,v_{1}) \xra{\tau_1} (l_1,v'_{1})\xrightarrow{a_2} \ldots \xrightarrow{a_{n}}(l_{n},v_n) \xra{\tau_{n}} (l_{n},v'_{n})$ be a run of a HA. The projection of $r$ onto the graph $\mathcal{G}_H$ is a timed word $\sigma = (t_1, a_1), (t_2, a_2), \dots, (t_n, a_n)$, where:
%    (i) For each discrete transition $(l_{i-1}, v'_{i-1}) \xrightarrow{a_i} (l_i, v_i)$ in $r$, there exists an edge $e = (l_{i-1}, a_i, l_i) \in E$ in the graph $\mathcal{G}_H$, such that $\{l_{i-1},l_i\}\in \emph{Loc}$ and {$a_i \in \Sigma$}.
%    (ii) The timestamp $t_i$ represents the cumulative time elapsed until the occurrence of the $i$-th discrete transition, defined as $t_i = \sum_{j=0}^{i-1} \tau_j$ for all $i \in \llbracket 1..n\rrbracket$. $\Box$
% \end{definition}

%\sr{We can just write this as $a_i \in \Sigma$??}

\subsection{Languages and Properties as Hybrid Automata}

Let $\mathbb{R}^{\geq 0}$ denote the set of non-negative real numbers, and $\Sigma$ a finite alphabet of actions. A pair ($t, a$) $\in$ ($\mathbb{R}^{\geq 0}\times \Sigma$) is called an \emph{event}, where $t \in \mathbb{R}^{\geq 0}$ is the time-stamp at which the event appears.
%A \emph{timed-word} over $\Sigma$ is a finite (respectively, infinite) sequence of events over $(\mathbb{R}^{\geq 0}\times \Sigma)^*$ (respectively, $(\mathbb{R}^{\geq 0}\times \Sigma)^\omega$).
A \emph{timed-word} over $\Sigma$ is a finite sequence of events over $(\mathbb{R}^{\geq 0}\times \Sigma)^*$.
%\sr{should be $(\mathbb{R}^{\geq 0}\times \Sigma)^\omega$) ?}.
We consider a timed-word $\sigma$ = ($t_1 , a_1$)($t_2, a_2$)$\cdots$($t_n, a_n$). For $i\in[1, n]$, $t_i$ is the time-stamp of the corresponding event, and $t_1$ is the time elapsed before the first action $a_1$.
Note that even though the alphabet ($\mathbb{R}^{\geq 0}\times\Sigma$) is infinite in this case, the notations defined above (related to length, concatenation, prefix, etc.) naturally extend to timed words.
A \emph{timed-language} is any set $\mathcal{L}$ $\subseteq$ $(\mathbb{R}^{\geq 0}\times \Sigma)^*$.
%\sr{should be $(\mathbb{R}^{\geq 0}\times \Sigma)^*$) ?}.
%Respectively, for the infinite sequence of events $\mathcal{L}^\omega$ $\subseteq$ $(\mathbb{R}^{\geq 0}\times \Sigma)^\omega$ denotes the language of infinite timed-words.

\begin{definition} [Projection of a Run on a Timed word]
   Let $r = (l_{0},v_{0}) \xra{\tau_0}(l_{0},v'_0) \xrightarrow{a_1} (l_1,v_{1}) \xra{\tau_1} (l_1,v'_{1})\xrightarrow{a_2} \ldots \xrightarrow{a_{n}}(l_{n},v_n) \xra{\tau_{n}} (l_{n},v'_{n})$ be a run of a HA. The projection of $r$ onto a timed word $\sigma = (t_1, a_1), (t_2, a_2), \dots, (t_n, a_n)$, where: for each discrete transition $(l_{i-1}, v'_{i-1}) \xrightarrow{a_i} (l_i, v_i)$ in $r$, there exists an event $e = (t_i,a_i)$, such that {$a_i \in \Sigma$} and the timestamp $t_i$ represents the cumulative time elapsed until the occurrence of the $i$-th discrete transition, defined as $t_i = \sum_{j=0}^{i-1} \tau_j$ for all $i \in \llbracket 1..n\rrbracket$. $\Box$
\end{definition}

%A trace of a run $r$ in a hybrid automaton is defined as a timed word $(\delta_1, a_1)$, $(\delta_2, a_2)$, $\cdots$, $(\delta_n, a_n)$. For each $i \in [1, n]$, the time-stamp $\delta_i = \sum_{j=0}^{i-1} \tau_j$ represents the absolute time at which the discrete action $a_i$ occurs, where $\tau_j$ is the dwell-time spent in location $l_j$ (refer to Def.~\ref{def:run}). Conceptually, a trace is a projection of the run $r$ onto the underlying graph $\mathcal{G}_H$ of the hybrid automaton HA, augmented with the specific time-stamps of the discrete transitions.

\noindent In this work, we consider finite runs of HA. A Muller acceptance condition of HA is given by a collection $\mathcal{F} \subseteq 2^{Loc}$ of acceptance sets of locations \cite{DBLP:conf/hybrid/AlurCHH92}. A finite \emph{run} with final location $l_n$ is accepting when $\{l_n\} \in \mathcal{F}$. We define the language of a HA as follows:
% or $r_{\infty}\in \mathcal{F}$ for the set $r_{\infty}$ of locations that are visited infinitely often during $r$ (i.e, $r_{\infty}$ is the set \{$l$ $|$ $l=l_i$ for infinitely many $i\geq 0$\}). 

\begin{definition} [Language of Hybrid Automata] \label{def:langHA}
    The language accepted by a HA, denoted by $\mathcal{L}$(HA), is the set of all timed words generated by accepting runs of HA$:$
    $\mathcal{L}$(HA) $=$ $\{\sigma \mid \sigma$ is generated by an accepting $run$ of HA, $\sigma \in (\mathbb{R}^{\geq 0} \times \Sigma)^*\}$.
\end{definition}

%\rr{Does not seem right to call this a subclass of HA. It is an example of HA. For instance, we don't refer to DFAs with a sink state as a subclass of DFAs.}
%\noindent \textbf{Safety Hybrid Automata} 
\noindent In this paper, we focus on prefix-closed properties. We formally express these properties as \emph{safety hybrid automata}. A safety hybrid automaton is an HA with a sink location \emph{Sink} and a transition $\textit{fail}$ from each HA location to the \emph{Sink}.
%whenever a monitoring property is violated \sr{Just say that it is a unique non-accepting/trap location. We do not need to talk about monitoring?}. 
\emph{Sink} is a unique non-accepting/trap location where \emph{Flow} of the variables are zero, and the invariant set is \emph{true}. We define a safety hybrid automaton below: 
%\sr{does it extend? Perhaps is a sub-class?}
\begin{definition} [Safety Hybrid Automata]
    A safety hybrid automaton $\mathcal{S}$ = \big(\emph{Loc}$_\mathcal{S}$, \emph{Var}, \emph{Flow}, \emph{Init}, \emph{$\Sigma$}$_\mathcal{S}$, \emph{Edge}, \emph{Inv}, \emph{Sink}\big) is an extension of a hybrid automaton HA with the following changes: (i) \emph{Loc}$_\mathcal{S}$ = \emph{Loc} $\cup$ $\{\emph{Sink}\}$, (ii) \emph{$\Sigma$}$_\mathcal{S}$ = \emph{$\Sigma$} $\cup$ $\{fail\}$, and (iii) All other components are same as HA. (iv) All locations in \emph{Loc}$_\mathcal{S}$ except \emph{Sink} (i.e., \emph{Loc}$_\mathcal{S}\setminus \{\emph{Sink}\}$) are accepting locations. \emph{Sink} is a unique non-accepting (trap) location, and there are no outgoing transitions in \emph{Edge} from \emph{Sink} to a location in \emph{Loc}$_\mathcal{S}\setminus \{\emph{Sink}\}$.
    $\Box$
\end{definition}

\noindent Consider the thermostat automaton in Figure~\ref{fig:thermostat}, a property P such as \textit{"the temperature should always remain between 15 and 25"} can be represented as a safety hybrid automaton as shown in Figure~\ref{fig:safety-thermostat} where \textit{Fail-mode} represents a \emph{Sink} location, and the violation of the property are shown as red transitions. A state $(l,v)$ in $\mathcal{S}$ is an accepting state if $l$ is an accepting location and $v\in$ Inv($l$).
\begin{figure}[htbp]
    \captionsetup{font=scriptsize}
    \vspace{-10pt}
    \centering
    \begin{subfigure}[b]{0.49\textwidth}
        \centering
        \includegraphics[width=0.9\textwidth]{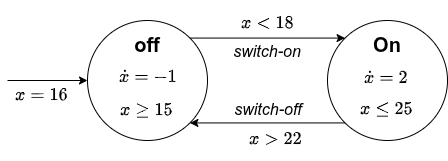}
        \caption{\scriptsize{A thermostat automaton.}}
        \label{fig:thermostat}
    \end{subfigure}
    \hfill
    \begin{subfigure}[b]{0.49\textwidth}
        \centering
        \includegraphics[width=0.9\textwidth]{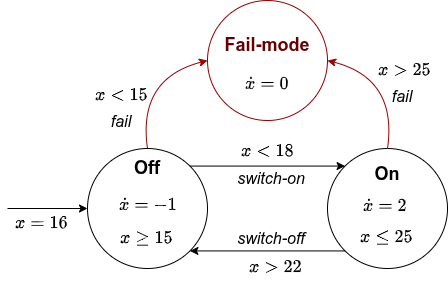}
        \caption{\scriptsize{A safety automaton for the property P.}}
        \label{fig:safety-thermostat}
    \end{subfigure}
    \caption{Thermostat automaton and corresponding safety thermostat automaton. }
    \label{fig:thermostat-example}
    \vspace{-10pt}
\end{figure}

% \begin{figure}[htbp]
%     %\captionsetup{font=scriptsize}
%     \centering
%     \includegraphics[width=0.5\textwidth]{images/safety-thermostat.drawio.png}
%     \caption{A parametric safety hybrid automaton for the property P.}
%     \label{fig:safety-thermostat}
% \end{figure}

%\todo[inline]{Example property defined as Safty Hybrid Automata}

\noindent\textbf{Language of Safety Hybrid Automata} The language accepted by a safety hybrid automaton $\mathcal{S}$ is a timed-language $\mathcal{L}(\mathcal{S}) \subseteq (\mathbb{R}^{\geq 0}\times \Sigma_{\mathcal{S}})^*$. %For a finite sequence of events $\mathcal{L}(\mathcal{S}) \subseteq (\mathbb{R}^{\geq 0}\times \Sigma_{\mathcal{S}})^*$.
\section{{Runtime Enforcement Problem}}

In this section, we formalize the RE problem for "black box" reactive systems {for properties specified as \emph{safety hybrid automata}}.
%\sr{for properties specified as ...}.
We assume that the system interacts with the environment through a special function call called \emph{ptick}, which is invoked exactly once during each reactive step. In each \emph{ptick}, the system receives inputs from the environment and outputs an event in $(\mathbb{R}^{\geq 0} \times \Sigma) $. These events are either a real event or a stutter {(consisting of a stutter action $d$)} if they are a continuation of a previous event.
We consider a proactive enforcement monitor that observes the inputs from the environment in each \emph{ptick} and corrects the outputs of the system whenever necessary. For a given safety property $\varphi \subseteq (\mathbb{R}^{\geq 0} \times \Sigma)^*$ specified as a \emph{safety hybrid automata}, we want to synthesize an enforcer.
%\red{The enforcer observes inputs from the environment to initially set the parameters of the automaton. Additionally, it also detects any abrupt changes in the inputs to identify any external attack on the system.}
%\rr{The introduction/motivating example gave an impression that we are proposing a continuous time monitoring and enforcement. Observation at pticks implies discrete time monitoring/enforcement.}
\subsection{{Problem Definition}}
%\sr{\color{magenta}{Sec title can be "Problem Definition"?}}
An enforcer for a property $\varphi \subseteq (\mathbb{R}^{\geq 0} \times \Sigma)^*$ can only edit the output sequence of events by suppressing, delaying, or inserting events whenever necessary. While suppressing and delaying an event in a \emph{ptick} is a reactive trait, the enforcer can also proactively insert an event before any \emph{ptick} (in the nick of time~\cite{DBLP:journals/jcs/BasinDH24}) whenever necessary if SuM does not produce an event on its own.
%The interval $\Delta$ between consecutive \emph{pticks} must satisfy the real-time constraint $\Delta > \delta_S + \delta_{S\leftrightarrow E} + \delta_E$~\cite{DBLP:conf/cav/HubletLBKT24}. Here, $\delta_S$ denotes the worst-case execution time required by the system to generate events after observing inputs from the environment. The term $\delta_{S\leftrightarrow E}$ represents the worst-case communication delay between the SuM and the enforcer $E_{\varphi}$, while $\delta_E$ accounts for the enforcer's maximum computational latency. This constraint ensures that the enforcement loop completes within a single control cycle.
At each \emph{ptick} $t \in \mathbb{N}\cup\{0\}$, the enforcer receives an event (a continuation of the incoming timed-word $\sigma \in (\mathbb{R}^{\geq 0} \times \Sigma)^*$) as input while observing the current environmental state $\mathcal{I}$ ($\mathcal{I}$ ranges over $\mathbb{R}^{|n|}$, where $n$ is number of variables), and monitor the property until next \emph{ptick} $t'$. It generates one or more events (a continuation of rectified output timed-word $\sigma' \in (\mathbb{R}^{\geq 0} \times \Sigma)^*$), such that the output satisfies the safety property (i.e., $\sigma' \models \varphi$) for all the time points $\tau \in [t,t']$.

At an abstract level an enforcer can be considered as a function that takes a timed word as input and emits a timed word. An enforcer for any given property $\varphi$ should satisfy the following constraints:
%We define the enforcement function as follows:

%\sr{An enforcer for any given property $\varphi$ should satisfy the following constraints:}
%At a \emph{ptick} $t \in \mathbb{N}\cup\{0\}$, input to the enforcer is a timed-word $\sigma \in (\mathbb{R}^{\geq 0} \times \Sigma^*)$, while it also observes the environmental inputs $\mathcal{I}$ and output is a timed-word $\sigma' \in (\mathbb{R}^{\geq 0} \times \Sigma^*)$, where $\sigma' \models \varphi$. 

%\sr{In addition to the timed word (input stream, enforcer also takes 3 other parameters as input. Mention and explain what are the other parameters it takes as input..}

%\sr{$\Sigma$ denotes only the set of output events? and $\mathcal{I} $ denotes inputs? Should it also be a stream of inputs?}

%\sr{Should it be $E_{\varphi}(\sigma,t,\tau,\mathcal{I})$: $(\mathbb{R}^{\geq 0} \times \Sigma)^* \times  \mathcal{I}^* \times (\mathbb{N}\cup\{0\}) \times \mathbb{R}^{\geq 0}  \to (\mathbb{R}^{\geq 0} \times \Sigma)^*$ ?}

\begin{definition} [Enforcer for $\varphi$] \label{def:enforcer}
 Given a property $\varphi \subseteq (\mathbb{R}^{\geq 0} \times \Sigma)^*$, an enforcer for $\varphi$ is a function {$E_{\varphi}(\sigma,t,\tau,\mathcal{I})$: $(\mathbb{R}^{\geq 0} \times \Sigma)^* \times (\mathbb{N}\cup\{0\}) \times \mathbb{R}^{\geq 0} \times \mathbb{R}^{|n|} \to (\mathbb{R}^{\geq 0} \times \Sigma)^*$} satisfying the following constraints:

 \begin{itemize}
     \item \textbf{\emph{(Soundness)}} $\forall \sigma \in (\mathbb{R}^{\geq 0} \times \Sigma)^*$, $\forall t,t' \in (\mathbb{N}\cup\{0\})$, $t\leq t'$, $\forall \tau \in [t,t'):$ $E_{\varphi}(\sigma,t,\tau,\mathcal{I}) \models \varphi$, where $\tau$ is a real-time between $t$ and $t'$.
     %\sr{exclude $t'$ in the interval}
     % \begin{align*}
     %     & \forall \sigma \in (\mathbb{R}^{\geq 0} \times \Sigma)^*,\ \forall t,t' \in (\mathbb{N}\cup\{0\}),\ t\leq t',\ \forall \tau \in [t,t'):\\
     %     & E_{\varphi}(\sigma,t,\tau,\mathcal{I}) \models \varphi.
     % \end{align*} where $\tau$ is a real-time between $t$ and $t'$.
 \end{itemize}
 
 %E_{\varphi}(\sigma,t,\tau,\mathcal{I}) = \sigma' \implies \sigma' \models \varphi, |\sigma'| \geq |\sigma|.
{
 \begin{itemize}
     \item \textbf{\emph{(Instantaneity)}} $\forall \sigma \in (\mathbb{R}^{\geq 0} \times \Sigma)^*$, $\forall t \in (\mathbb{N}\cup\{0\})$, $\forall \tau \in [t,t+1):$\\
     $E_{\varphi}(\sigma,t,\tau,\mathcal{I}) \geq |\sigma|$
     % \begin{align*}
     %    & \forall \sigma \in (\mathbb{R}^{\geq 0} \times \Sigma)^*, \forall t \in (\mathbb{N}\cup\{0\}),\ \forall \tau \in [t,t+1):\\
     %    & E_{\varphi}(\sigma,t,\tau,\mathcal{I}) \geq |\sigma|
     %    % & \forall \sigma, \sigma' \in (\mathbb{R}^{\geq 0} \times \Sigma)^*, E_{\varphi}(\sigma,t,\tau,\mathcal{I}) = \sigma', \ \forall t \in (\mathbb{N}\cup\{0\}),\ \forall \tau \in [t,t+1):\\
     %    % &\forall\ e(t_1,a)\in \sigma,\ \exists\ e'(t_2,a') \in \sigma'
     % \end{align*} %where $a, a' \in \Sigma$, $t_1,t_2 \in \mathbb{R}^{\geq 0}$, and $t_2\geq t_1$.
 \end{itemize}
 }

 \begin{itemize}
     \item \textbf{\emph{(Monotonicity)}} $\forall \sigma, \sigma' \in (\mathbb{R}^{\geq 0} \times \Sigma)^*$, $\forall t,t' \in (\mathbb{N}\cup\{0\}):$\\
     $\sigma \preceq \sigma'$, $t\leq t'$ $\implies$ $E_{\varphi}(\sigma,t,\tau,\mathcal{I})$ $\preceq$ $E_{\varphi}(\sigma',t',\tau',\mathcal{I}')$, where $\mathcal{I}$ and $\mathcal{I}'$ are environmental inputs at $t$ and $t'$, respectively. $\tau \in [t,t')$ and $\tau' \in [t',t'+1)$.
     % \begin{align*}
     %     & \forall \sigma, \sigma' \in (\mathbb{R}^{\geq 0} \times \Sigma)^*, \forall t,t' \in (\mathbb{N}\cup\{0\}):\\
     %     & \sigma \preceq \sigma', t\leq t' \implies E_{\varphi}(\sigma,t,\tau,\mathcal{I}) \preceq E_{\varphi}(\sigma',t',\tau',\mathcal{I}').
     % \end{align*} where $\mathcal{I}$ and $\mathcal{I}'$ are environmental inputs at $t$ and $t'$, respectively. $\tau \in [t,t')$ and $\tau' \in [t',t'+1)$.
 \end{itemize}
 
 \begin{itemize}
     \item \textbf{\emph{(Transparency)}} (i) $\forall \sigma \in (\mathbb{R}^{\geq 0} \times \Sigma)^*$, $\forall e \in (\mathbb{R}^{\geq 0} \times \Sigma )$, $\forall t, t' \in (\mathbb{N}\cup\{0\})$, $t \leq t'$, $\forall \tau \in [t,t'):$ $E_{\varphi}(\sigma,t,\tau,\mathcal{I})\cdot e \models \varphi$ $\implies$ $E_{\varphi}(\sigma\cdot e,t',\tau',\mathcal{I}')$ $=$ $E_{\varphi}(\sigma,t,\tau,\mathcal{I})\cdot e$, where $\tau'\in[t',t'+1)$.
       (ii) $\forall \sigma \in (\mathbb{R}^{\geq 0} \times \Sigma)^*$, $\forall t, t' \in (\mathbb{N}\cup\{0\})$, $t \leq t'$, $\forall \tau \in [t,t'):$ $E_{\varphi}(\sigma,t,\tau,\mathcal{I}) \models \varphi$ $\implies$ $E_{\varphi}(\sigma,t,\tau,\mathcal{I})$ $=$ $E_{\varphi}(\sigma,t,t,\mathcal{I})$. $\Box$
     % \begin{align*}
     %     & (i)\ \forall \sigma \in (\mathbb{R}^{\geq 0} \times \Sigma)^*, \forall e \in (\mathbb{R}^{\geq 0} \times \Sigma ), \forall t, t' \in (\mathbb{N}\cup\{0\}), t \leq t', \forall \tau \in [t,t'):\\
     %     & \ E_{\varphi}(\sigma,t,\tau,\mathcal{I})\cdot e \models \varphi \implies E_{\varphi}(\sigma\cdot e,t',\tau',\mathcal{I}') = E_{\varphi}(\sigma,t,\tau,\mathcal{I})\cdot e,\\ & where\ \tau'\in[t',t'+1).\\
     %     & (ii)\ \forall \sigma \in (\mathbb{R}^{\geq 0} \times \Sigma)^*, \forall t, t' \in (\mathbb{N}\cup\{0\}), t \leq t', \forall \tau \in [t,t'):\\
     %     & E_{\varphi}(\sigma,t,\tau,\mathcal{I}) \models \varphi \implies E_{\varphi}(\sigma,t,\tau,\mathcal{I}) = E_{\varphi}(\sigma,t,t,\mathcal{I}). \Box
     % \end{align*} %where $\mathcal{I}$ and $\mathcal{I}'$ are environmental inputs at $t$ and $t'$, respectively.
 \end{itemize}

\end{definition}

\paragraph{Soundness.} It expresses that for any input sequence observed and at any given time, the sequence of events released as output by the enforcer must satisfy the property $\varphi$.

\paragraph{Instantaneity.} It expresses that for any observed input sequence $\sigma$, and tick $t$, the length of output produced $E_{\varphi}(\sigma,t,\tau,\mathcal{I})$ should be greater than $\sigma$ i.e., $|E_{\varphi}(\sigma,t,\tau,\mathcal{I})|\geq |\sigma|$. This is because the enforcer can proactively insert events in between two $pticks$ whenever necessary. Instantaneity condition ensures that the timestamps of events in $E_{\varphi}(\sigma,t,\tau,\mathcal{I})$ reflect the time at which the corresponding events are performed. Whenever the enforcer receives a new event, it has to react instantaneously and has to produce an output event immediately.

\paragraph{Monotonicity.} The enforcer satisfies the monotonicity property: for any input $\sigma$ and its extension $\sigma'= \sigma \cdot e$, the output produced for the extension must be a continuation of the output produced for $\sigma$. Any output generated by $E_{\varphi}$ at a given time step is preserved in all future iterations, reflecting the irreversible nature of real-time event streams.

\paragraph{Transparency.} %Transparency means that the enforcer will not unnecessarily edit or insert a new event. The first constraint expresses that any new input event $e$ will be simply forwarded by the enforcer if the continuation of the enforcer output with the event $e$ satisfy the property $\varphi$, i.e., if $E_{\varphi}(\sigma,t,\tau,\mathcal{I})\cdot e \models \varphi$, then $E_{\varphi}(\sigma\cdot e,t',\tau',\mathcal{I}') = E_{\varphi}(\sigma,t,\tau,\mathcal{I})\cdot e$. The second constraint specifies that it should not unnecessarily insert an event between the interval $[t,t']$ if $\sigma'$ satisfies $\varphi$ for the duration.
Transparency ensures that the enforcer performs the minimal necessary intervention, avoiding redundant edits or insertions. This property is governed by two primary constraints:
(\emph{Identity Preservation}) If the current output trace, extended by a new input event $e$, satisfies the safety property $\varphi$ (i.e., $E_{\varphi}(\sigma, t, \tau, \mathcal{I}) \cdot e \models \varphi$), the enforcer must forward $e$ without modification. Formally: $E_{\varphi}(\sigma \cdot e, t', \tau', \mathcal{I}') = E_{\varphi}(\sigma, t, \tau, \mathcal{I}) \cdot e$.
(\emph{Non-Intrusion}) The enforcer is prohibited from inserting unnecessary events during the interval {$[t, t')$} if the existing output sequence $\sigma'$ remains compliant with $\varphi$ throughout that duration.
% \begin{itemize}
%     \item Identity Preservation: If the current output trace, extended by a new input event $e$, satisfies the safety property $\varphi$ (i.e., $E_{\varphi}(\sigma, t, \tau, \mathcal{I}) \cdot e \models \varphi$), the enforcer must forward $e$ without modification. Formally: $E_{\varphi}(\sigma \cdot e, t', \tau', \mathcal{I}') = E_{\varphi}(\sigma, t, \tau, \mathcal{I}) \cdot e$.

%     \item Non-Intrusion: The enforcer is prohibited from inserting unnecessary events during the interval {$[t, t')$} if the existing output sequence $\sigma'$ remains compliant with $\varphi$ throughout that duration.
% \end{itemize}

Together, these constraints ensure that the enforcer acts as an identity function whenever the system behaves safely.

\ignore{
\textcolor{blue}{
\paragraph{Soundness and Instantaneity.} The enforcer $E_{\varphi}$ ensures that the sequence of events it outputs satisfies the property $\varphi$. Basin et al.~\cite{DBLP:journals/jcs/BasinDH24} state two conditions for achieving soundness: (i) the SuM and enforcer must be synchronized, and (ii) the enforcer must be fast enough to keep up with the real-time system behavior. These conditions also apply in our model.
Soundness condition ensures that the order of events in $\sigma'$ observed by $E_{\varphi}$ reflects the order of events in $\sigma$. However, the length of $\sigma'$ can be greater than $\sigma$, i.e., $|\sigma'|\geq |\sigma|$, as the enforcer can proactively insert events in between two $pticks$ whenever necessary. Instantaneity condition ensures that the timestamps of events in $\sigma'$ reflect the time at which the corresponding events are performed. Whenever the enforcer receives a new event, it has to react instantaneously and has to produce an output event immediately} \sr{Unclear..to discuss.}. \sr{There is only one constraint expressed currently under soundness?  Also, mainly the output produced by the enforcer must satisfy the property being monitored $\varphi$..}

%The interval $\Delta$ between two \emph{pticks} must satisfy the real-time condition $\Delta > \delta_S + 2\delta_{S\leftrightarrow E} + \delta_E$~\cite{DBLP:conf/cav/HubletLBKT24}, where where $\delta_S$ is the worst-case time needed by the system to create events before performing observable actions and process the enforcer’s reactions, $\delta_{S\leftrightarrow E}$ is the worst-case communication time between the SuM and the enforcer $E_{\varphi}$, and $\delta_E$ is the worst-case latency of the enforcer.

\paragraph{Monotonicity.} The enforcer satisfies the monotonicity property: for any input $\sigma$ and its extension $\sigma'= \sigma \cdot e$, the output produced for the extension must be a continuation of the output produced for $\sigma$. Any output generated by $E_{\varphi}$ at a given time step is preserved in all future iterations, reflecting the irreversible nature of real-time event streams.

\paragraph{Transparency.} %Transparency means that the enforcer will not unnecessarily edit or insert a new event. The first constraint expresses that any new input event $e$ will be simply forwarded by the enforcer if the continuation of the enforcer output with the event $e$ satisfy the property $\varphi$, i.e., if $E_{\varphi}(\sigma,t,\tau,\mathcal{I})\cdot e \models \varphi$, then $E_{\varphi}(\sigma\cdot e,t',\tau',\mathcal{I}') = E_{\varphi}(\sigma,t,\tau,\mathcal{I})\cdot e$. The second constraint specifies that it should not unnecessarily insert an event between the interval $[t,t']$ if $\sigma'$ satisfies $\varphi$ for the duration.
Transparency ensures that the enforcer performs the minimal necessary intervention, avoiding redundant edits or insertions. This property is governed by two primary constraints:
\begin{itemize}
    \item Identity Preservation: If the current output trace, extended by a new input event $e$, satisfies the safety property $\varphi$ (i.e., $E_{\varphi}(\sigma, t, \tau, \mathcal{I}) \cdot e \models \varphi$), the enforcer must forward $e$ without modification. Formally: $E_{\varphi}(\sigma \cdot e, t', \tau', \mathcal{I}') = E_{\varphi}(\sigma, t, \tau, \mathcal{I}) \cdot e$.

    \item Non-Intrusion: The enforcer is prohibited from inserting unnecessary events during the interval \textcolor{blue}{$[t, t')$} if the existing output sequence $\sigma'$ remains compliant with $\varphi$ throughout that duration.
\end{itemize} Together, these constraints ensure that the enforcer acts as an identity function whenever the system behaves safely.
}

\subsection{Algorithm}
In this section, we present the algorithm for the Enforcer $E_{\varphi}$ for a given safety property $\varphi$. We assume that $\varphi$ is formally defined by a safety hybrid automaton $\mathcal{S}$. Algorithm~\ref{algo:enforcer} implements an online reactive interface between the SuM and its environment. The algorithm accepts the automaton $\mathcal{S}$, the synchronization interval \emph{pticks}, and the environmental input set $\mathcal{I}$ as parameters. For an incoming online sequence of system events $\sigma$, where one event is produced at each synchronization step, the algorithm synthesizes an edited output sequence $\sigma'$ that satisfies the property $\varphi$.

The function $\textit{Initialize} (\mathcal{I})$ initializes the hybrid automaton $\mathcal{S}$ in Line 1. Subsequently, Line 2 defines the execution context by setting the current state ($\textit{curState}$) to the automaton's initial state, resetting the simulation time ($\textit{curTime}$) to $0.0$, and nullifying the $\textit{delayedTransition}$ variable—which tracks a pending transition from the previous synchronization step. During each \emph{ptick}, the algorithm receives an output event $e$ as a continuation of the sequence $\sigma$. It calculates the system's dwell time on Line 5 and invokes the $\textit{timeElapse}()$ function on Line 7 to verify that the control remains at the current location and that all invariants are satisfied. The algorithm then proceeds to edit the event $e$ by performing the following checks:
\textbf{Case 1 (invariant not violated and guard satisfied)} If control remains within the current location and invariant conditions are not violated for the entire dwell time, while guard conditions of the action in the event are satisfied at that time, the algorithm outputs the event without modifying it in line $38$ (causing). It then \emph{goto} step $7$ to monitor the remaining part of the \emph{ptick}.
% \begin{cse} [invariant not violated and guard satisfied]
%     If control remains\\ within the current location and invariant conditions are not violated for the entire dwell time, while guard conditions of the action in the event are satisfied at that time, the algorithm outputs the event without modifying it in line $38$ (causing). It then \emph{goto} step $7$ to monitor the remaining part of the \emph{ptick}.
% \end{cse}
\textbf{Case 2 (invariant not violated, but guard not satisfied)} If control is within the current location and invariant conditions are not violated, however, the guard conditions of the action in the event are not satisfied at that time. The algorithm suppresses the event (suppressing), while generating a \emph{stutter} event that allows continuation of the previous event. It puts the suppressed event in the delayed event in line $45$ to process the event at a later time if possible. It goes to step $7$ for the further processing of the remaining time of the \emph{ptick}.
% \begin{cse} [invariant not violated, but guard not satisfied]
%     If control is\\ within the current location and invariant conditions are not violated, however, the guard conditions of the action in the event are not satisfied at that time. The algorithm suppresses the event (suppressing), while generating a \emph{dummy} event that allows continuation of the previous event. It puts the suppressed event in the delayed event in line $45$ to process the event at a later time if possible. It goes to step $7$ for the further processing of the remaining time of the \emph{ptick}.
% \end{cse}
\textbf{Case 3 (invariant violated and delayed event possible)} If invariant conditions are violated within the current location, the algorithm checks for an enabled delayed action, which means guard conditions of the action are satisfied at that time. The action, along with its delayed time, is output as a delayed event in line $14$ (delaying). The control goes to step $7$ for the monitoring of the remaining time of the \emph{ptick}.
% \begin{cse} [invariant violated and delayed event possible]
%     If invariant conditions are violated within the current location, the algorithm checks for an enabled delayed action, which means guard conditions of the action are satisfied at that time. The action, along with its delayed time, is output as a delayed event in line $14$ (delaying). The control goes to step $7$ for the monitoring of the remaining time of the \emph{ptick}.
% \end{cse}
\textbf{Case 4 (invariant violated and delayed event not enabled)} If invariant conditions are violated within the current location, and the guard is not satisfied at that time for a delayed action, the algorithm finds an alternate action that is enabled at that time. It outputs the alternate action along with the new time as an inserted event in line $22$ (suppressing-inserting). The control goes to step $7$.
% \begin{cse} [invariant violated and delayed event not enabled]
%    If invariant conditions are violated within the current location, and the guard is not satisfied at that time for a delayed action, the algorithm finds an alternate action that is enabled at that time. It outputs the alternate action along with the new time as an inserted event in line $22$ (suppressing-inserting). The control goes to step $7$.
% \end{cse}
\textbf{Case 5 (invariant violated and no delayed event exists)} If invariant is violated and no delayed event exists, the algorithm finds an action among the enabled actions at that time to insert as a new event in line $31$ (inserting). The control goes to step $7$ to monitor the remaining time of the \emph{ptick}.
% \begin{cse} [invariant violated and no delayed event exists]
%    If invariant is violated and no delayed event exists, the algorithm finds an action among the enabled actions at that time to insert as a new event in line $31$ (inserting). The control goes to step $7$ to monitor the remaining time of the \emph{ptick}.
% \end{cse}

\begin{algorithm}
\scriptsize
\SetKwInOut{Input}{input}\SetKwInOut{Output}{output}
\Input{A safety hybrid automaton $\mathcal{S}$, the synchronization $pticks$, an output event sequence $\sigma$, and the environmental inputs $\mathcal{I}$.}
\Output{A sequence of edited output events $\sigma'$.}
\DontPrintSemicolon
\caption{\textsf{Enforcer} $E_{\varphi}$}
\label{algo:enforcer}
\BlankLine
\LinesNumbered
%\SetAlgoLined
%{\{$initState$, $params$\} $\gets$ $\textit{Input}_{\varphi}(\mathcal{I})$; \scriptsize{\tcc*{Initialize $\mathcal{S}$ to standard hybrid automaton.}} }
{\{$initState$\} $\gets$ $\textit{Initialize}(\mathcal{I})$; \scriptsize{\tcc*{Initialize $\mathcal{S}$.}} }
Set $curState \gets initState$, $curTime \gets 0.0$, $delayedTransition \gets null$; \\
%\scriptsize{\tcp*[l]{The function $\textit{Edit}_\varphi$ begins:}}
\For{each $ptick$ t}{
  $nxtPtick \gets t + 1$;\\
  Get next event $e$: $\{eventTime, transition\} \gets \sigma$;\\
  %$nxtEventTime \gets eventTime$;\\
  {Compute $deltaTime$ = ($eventTime - curTime$); \scriptsize{\tcc*{Compute the time to elapse.}} }
  {\{$invariantOk$, $curState$, $curTime$\} $\gets$ $timeElapse(curState,deltaTime)$; \scriptsize{\tcc*{$timeElapse()$ returns false the moment invariant is violated.}} }
  \eIf{curTime $<$ nxtPtick}
     {
     \eIf{invariantOk $\neq$ True}
     {
        \eIf{delayedTransition $\neq$ null}
        {
           $e' \gets \{curTime, delayedTransition\}$
           {$transPoss \gets delayedTransitionPossible()$; \scriptsize{\tcc*{Returns $True$ if guard of the delayed transition satisfied at that time.}} }
           \eIf{transPoss == True}
           {
              $\sigma' \gets \sigma'.append(e')$;\\
              {Output $\sigma'$; \scriptsize{\tcc*{A delayed event is outputted (delaying).}} }
              $delayedTransition \gets null$;\\
              {Compute $deltaTime$ = ($nxtPtick - curTime$); \scriptsize{\tcc*{Time elapse for the remaining part of the current $ptick$.}}}
              $goto$ Step 7;
           }
           {
             {Find a transition $a \in enabledTransitions(curState)$; \scriptsize{\tcc*{Find a transition in current state whose guard is satisfied.}} }
             $e' \gets \{curTime, a\}$;\\
             $\sigma' \gets \sigma'.append(e')$;\\
             {Output $\sigma'$; \scriptsize{\tcc*{Event is inserted as invariant is violated (delayed event is suppressed) (suppressing-inserting).}} }
             $delayedTransition \gets null$;\\
             Compute $deltaTime$ = ($nxtPtick - curTime$);\\
             $goto$ Step 7;
           }
        }
        {
           {Find a transition $a \in enabledTransitions(curLoc)$; \scriptsize{\tcc*{Find a transition in current location (no delayed event exists).}} }
           $e' \gets \{curTime, a\}$;\\
           $\sigma' \gets \sigma'.append(e')$;\\
           {Output $\sigma'$; \scriptsize{\tcc*{Invariant violation, event inserted (inserting).}} }
           Compute $deltaTime$ = ($nxtPtick - curTime$);\\
           $goto$ Step 7;
        }
     }
     {
        \eIf{transition $\in$ enabledTransition(curState)}
        {
           $\sigma' \gets \sigma'.append(e)$;\\
           {Output $\sigma'$; \scriptsize{\tcc*{The event is executable at that time.}} }
           Compute $deltaTime$ = ($nxtPtick - curTime$);\\
           $goto$ Step 7;
        }
        {
           {Suppress-delay event: $e' \gets d$; \scriptsize{\tcc*{$d$ is a stutter event that allows continuation of the previous event.}} }
           $\sigma' \gets \sigma'.append(e')$\\
           {Output $\sigma'$; \scriptsize{\tcc*{The event is suppressed (suppressing).}} }
           $delayedTransition \gets transition$;\\
           Compute $deltaTime$ = ($nxtPtick - curTime$);\\
           $goto$ Step 7;
        }
     }
    }
    {$continue$;
    }
}
\end{algorithm}

\subsection{Enforceability}
%\rr{Re vs. Rv discussion here reads redundant and out of place. This discussion can go once in introduction.}
%\rr{Can we be more specific and say the operations required in the algorithm. One of the operation is timeElapse, which is not a decision question but a computation question. We need to check under what conditions is timeElapse computable, exactly or approximately.}
We now examine the enforceability of our proposed approach.
The enforcer must synthesize a rectified event sequence by identifying an admissible set of enabled transitions. Consequently, the execution of Algorithm~\ref{algo:enforcer} fundamentally depends on solving underlying state reachability and membership problems.
For example, the timeElapse() function, which computes the continuous state evolution over a time interval $\Delta t$, and the enabledTransitions() function, which evaluates guard conditions against the current state.
Because these problems are solvable for certain subclasses of hybrid automata (e.g, timed automata, initialized rectangular automata, and time-bounded monotonic linear hybrid automata~\cite{DBLP:conf/stoc/HenzingerKPV95}), the applicability of our framework is restricted to those subclasses.
In our model, a property $\varphi$ is deemed \textit{non-enforceable} if the set of enabled actions is empty at a critical junction where a corrective action is required to prevent a safety violation.
%Because these decision problems \rr{What is the decision problem?}are solvable for certain subclasses of hybrid automata, the applicability of our framework is restricted to those subclasses} initialized rectangular automata, and time-bounded monotonic linear hybrid automata~\cite{DBLP:conf/stoc/HenzingerKPV95}.
We define enforceability as follows:

\begin{definition} [Enforceability~\cite{DBLP:journals/tecs/PinisettyRSATH17}]
    Let $\varphi \subseteq (\mathbb{R}^{\geq 0} \times \Sigma)^*$ be a property, we say that $\varphi$ is enforceable iff an enforcer $E_{\varphi}$ for $\varphi$ exists according to Definition~\ref{def:enforcer}.
\end{definition}

\begin{definition} [Condition for Enforceability] \label{def:enforceabilityCondition}
    Consider a property $\varphi$ that is defined as a safety hybrid automaton $\mathcal{S}$ = \big(\emph{Loc}$_\mathcal{S}$, \emph{Var}, \emph{Flow}, \emph{Init}, \emph{$\Sigma$}$_\mathcal{S}$, \emph{Edge}, \emph{Inv}, \emph{Sink}\big), property $\varphi$ is enforceable iff the following condition holds:
    \begin{align*}
        & \forall l \in \emph{Loc}_\mathcal{S} \setminus \emph{Sink}, \forall v \in \emph{Inv}(l) \implies \exists a \in \Sigma_{\mathcal{S}}\setminus \{fail\} : (l,v)\xrightarrow{a} (l',v'),\\
        & where\ l'\in \emph{Loc}_\mathcal{S} \setminus \emph{Sink},\ and\ v'\in \emph{Inv}(l')
    \end{align*}
\end{definition}

%\rr{What is an accepting state? Is that mentioned earlier?}
%\rr{Is a state reachable from itself? If so, then this definition will be satisfied for any HA. What about dummy/stutter actions?}
\noindent The enforceability condition expresses that for a property $\varphi$ defined as a safety hybrid automaton $\mathcal{S}$ where all locations in \emph{Loc}$_\mathcal{S}\setminus$\emph{Sink} are accepting locations, $E_{\varphi}$ according to Definition~\ref{def:enforcer} exists iff an accepting state $(l',v')$ is reachable from every non-violating state $(l,v)$ where $l$ and $l'$ are accepting locations in \emph{Loc}$_\mathcal{S}\setminus$\emph{Sink}, and $v'\in$ \emph{Inv}($l'$).

%%%%%%%%%%%%%%%%%%%%%%%%%%%%%%%%%
\begin{definition}[$\textit{E}^*_{\varphi}$]
    %\sr{Define $E'_\varphi$ based on the algo}\\
    Consider an enforceable safety property $\varphi \subseteq (\mathbb{R}^{\geq 0} \times \Sigma)^*$. We define the function $E^*_{\varphi}$: $(\mathbb{R}^{\geq 0} \times \Sigma)^*$ $\to$ $(\mathbb{R}^{\geq 0} \times \Sigma)^*$ as follows: Let $\sigma$ $=$ $e_1\cdot e_2\cdots e_m$ be a timed-word received by Algorithm~\ref{algo:enforcer}. Then we let $E^*_{\varphi}(\sigma)$ $=$ $e_1'\cdot e_2'\cdots e_n'$, where $n\geq m$. $e_t'$ is an event output by Algorithm~\ref{algo:enforcer}. If $E^*_{\varphi}(\sigma) = \sigma'$, then the following holds:
    %\\Given a property $\varphi \subseteq (\mathbb{R}^{\geq 0} \times \Sigma)^*$, for $\forall \sigma,\sigma'$ $\in$ $(\mathbb{R}^{\geq 0} \times \Sigma)^*$, and $\forall t,t'$ $\in (\mathbb{N}\cup\{0\})$, where $t \leq t'$, and $\tau = [t,t')$, if $\textit{Edit}_{\varphi}(\sigma,t,\tau)$ $=$ $\sigma'$, then for $\forall e \in (\mathbb{R}^{\geq 0} \times \Sigma)$:
    \begin{itemize}
       \item[(i)] $\sigma'\cdot e$ $\models$ $\varphi$ $\implies$ $\textit{E}^*_{\varphi}(\sigma'\cdot e)$ $=$ $\sigma'\cdot e$.
       \item[(ii)] $\sigma'\cdot e$ $\not\models$ $\varphi$ $\implies$ $\textit{E}^*_{\varphi}(\sigma'\cdot e)$ $=$ $\sigma'\cdot \emph{stutter}$, where \emph{stutter} is a stutter event, \emph{stutter} $\in$ $(\mathbb{R}^{\geq 0} \times \{d\})$ allowing continuation of a previous event.
       \item[(iii)] $(\exists$ $\tau' \in (t',t'+1)$ $|$ $\sigma'$ $\not\models$ $\varphi)$ $\implies$ $\textit{E}^*_{\varphi}(\sigma')$ $=$ $\sigma'\cdot e$ $|$ $\sigma'\cdot e$ $\models$ $\varphi$, where $t'$ and $t'+1$ are two consecutive $pticks$, and $\tau\in [t',t'+1)$ (inserting events when invariants are violated). $\Box$
   \end{itemize}
\end{definition} %Where for every output event $e$, the enforcer produces rectified output events in each reactive cycle $ptick$. It allows the event to output when $e$ satisfies $\varphi$, in that case, $e'=e$. It suppresses or delays the event $e$ if it does not satisfy $\varphi$. It generates a dummy event $d$ in those cases to allow continuations of the previous events. The enforcer can also insert events between two consecutive $pticks$ as needed.

%\sr{We can define a proposition as follows.. if possible give a skecteh of correctness proof in an appendix !!}
\begin{proposition} [Correctness of the enforcement algorithm] \label{prop:prop2}
Given any \emph{safety hybrid automaton} $\mathcal{S}$ for the property $\varphi$ that satisfies the condition in Definition~\ref{def:enforceabilityCondition}, the function $E^*_\varphi$ defined above is an enforcer for $\varphi$, i.e., it satisfies all the constraints as per Definition~\ref{def:enforcer}.
\end{proposition}

%\sr{Add a line that the proof is in the Appendix..}
\noindent A brief proof sketch of Proposition~\ref{prop:prop2} is given in Appendix~\ref{appendix2}.

\section{Case Study}

Modern autonomous vehicles offer many different safety and comfort functions. Most of them are nowadays realized in software running on a bunch of electronic control units with connected actuators and sensors. For example, an adaptive cruise control system of a car maintains the speed of the vehicle without human intervention. However, these systems are susceptible to sensor errors or adversarial attacks, which can lead to catastrophic outcomes. A runtime monitoring and enforcement system can detect such situations and generate commands in accordance with the security policy to prevent safety violations.

In this section, we present a case study of our approach on \emph{adaptive cruise controller} (ACC) system~\cite{DBLP:conf/tacas/BogomolovFGH17} of an autonomous vehicle, which is also a well-known domain in hybrid systems. An ACC is a reactive system; it works in synchronization with its environment. In each synchronization step, it takes in the environmental inputs and produces commands/events for the system. In this work, we take into account some real-world specifications and safety requirements for modern cars as presented in~\cite{DBLP:conf/asm/HoudekR20}, like maximum allowable velocity on a high-speed lane, minimum safety distance requirements, etc. A description of the domain and a scenario is presented in Section~\ref{sec:motive}, where we assume the primary controller of the system is faulty, which can jeopardize safety requirements. An example safety requirement is shown in Figure~\ref{fig:car-cruise}, where the current car needs to maintain a safe distance of $10$ meters with the preceding car.

% \begin{figure}[htbp]
%     %\captionsetup{font=scriptsize}
%     \centering
%     \includegraphics[width=0.5\textwidth]{images/car-cruise.drawio.png}
%     \caption{Schematic diagram of the ACC where the safety requirement is to maintain the safe distance $S_d > 10$. $V_C$ and $V_f$ represent the velocity of the current car and the front car, respectively.}
%     \label{fig:car-cruise}
% \end{figure}

\begin{figure}[htbp]
    %\captionsetup
    \vspace{-20pt}
    \centering
    \begin{subfigure}[b]{0.49\textwidth}
        \centering
        \includegraphics[width=0.95\textwidth]{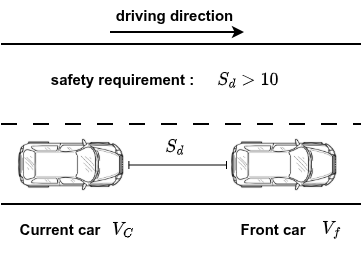}
        \caption{The safety requirement is to maintain the safe distance $S_d > 10$.}
        \label{fig:car-cruise}
    \end{subfigure}
    \hfill
    \begin{subfigure}[b]{0.49\textwidth}
        \centering
        \includegraphics[width=0.95\textwidth]{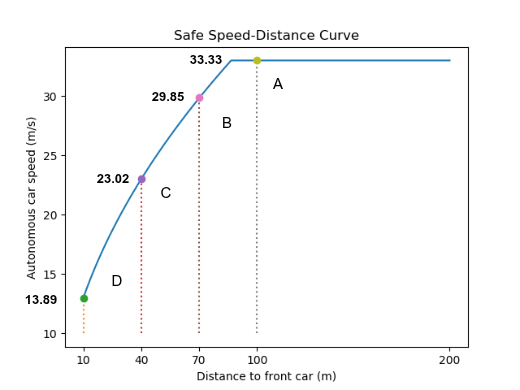}
        \caption{The speed-distance curve.}
        \label{fig:vel-dist}
    \end{subfigure}
    \caption{Safety requirement and Speed-distance curve}
    \label{fig:car-safety}
    \vspace{-20pt}
\end{figure}

%For a set of expected properties of the system to be monitored, one major challenge in building a runtime enforcement system is to write an appropriate monitor specification~\cite{DBLP:conf/kbse/LegunsenHXRM16}.
A central challenge in developing runtime enforcement frameworks is the formalization of monitor specifications that accurately capture the desired system properties~\cite{DBLP:conf/kbse/LegunsenHXRM16}.
In this work, we provide a monitor specification for runtime enforcement of the ACC system. The car movement control model is formalized with an automaton language. For a safety requirement, as shown in Figure~\ref{fig:car-cruise}, the maintenance of the speed does not only depend on the desired speed but also on the vehicle ahead.
Furthermore, the car's motion involves both discrete control modes (e.g., acceleration, deceleration, braking) and continuous features (e.g., car speed and distance to the preceding car) in its behaviour.
%Therefore, it is hard to specify with a static formal language. We use \emph{parametric hybrid automata} as a modeling language, an extension of hybrid automata that introduces parametric expressions for flow, transition conditions, and invariants. For the safety requirement in Figure~\ref{fig:car-cruise}, the parametric hybrid automaton modeling of the domain is shown in Figure~\ref{fig:monitor-pha}.
We use \emph{safety hybrid automata} as a modeling language for the safety requirement in Figure~\ref{fig:car-cruise}. The safety hybrid automaton modeling of the domain is shown in Figure~\ref{fig:monitor-ha}.
%Depending on the system's discrete state changes, the ACC model has five control modes. Apart from that, there is also a \emph{Sink} (not shown in the figure for simplicity of the diagram) that represents a \emph{Failure-mode}. From each of the five modes shown in figure, there is a transition to the \emph{Sink} whenever the safety property is violated. The guard, flow, and invariant conditions are parametric in the automaton, which can be instantiated to a standard hybrid automaton by initializing the parameters.
Our enforcement framework for the ACC system uses a hybrid automaton with five discrete control modes and an implicit \emph{Sink} location (omitted from the figure for clarity) that represents a safety failure. From each operational location, a transition to the Sink is enabled whenever the safety property is violated.
%Crucially, the automaton's dynamics—including flows, guards, and location invariants—are parameterized. This parametric nature enables flexible synthesis of safety monitors, as the model can be instantiated as a concrete hybrid automaton by initializing its parameters for a given scenario.

% \begin{figure}[htbp]
%     %\captionsetup{font=scriptsize}
%     \centering
%     \includegraphics[width=0.75\textwidth]{images/monitoring-pha.drawio.png}
%     \caption{The Parametric Hybrid Automaton where $a$, $v$, and $d$ represent acceleration, velocity, and the distance from the front car, respectively. $u$ represents the velocity of the front car. $y$, $z$, $v_f$, $d_1$, $d_2$, $d_3$, and $s_d$ are the parameters to the automaton.}
%     \label{fig:monitor-pha}
% \end{figure}

Initially, the car is in \emph{Init Mode}. We compute a \emph{speed-distance curve} for a given scenario to generate the monitoring specifications. For example, given a safety distance $S_d = 10m$ (meters), a maximum speed of the current car $V_C = 120kmph$ (kilometers per hour), a constant speed of the front car $V_f = 50kmph$, an initial distance from the front car $d = 200m$, and a maximum braking limit $K = -6m/s^2$ (meters per second$^2$).
%\rr{What is braking limit? Is it deceleration? Should this be negative?}.
We compute the \emph{speed-distance curve} (as shown in Figure~\ref{fig:vel-dist}) with a reverse iterative algorithm based on the following equation: $V_{safe} = \sqrt{V_f^2 + 2\times K \times (d - S_d)}$,
% \begin{align}
%     V_{safe} = \sqrt{V_f^2 + 2\times K \times (d - S_d)}
% \end{align}
where $V_{safe}$ represents the maximum allowable speed of the car at a point in the speed-distance curve.
We partition the operational curve into four distinct segments: $A, B, C,$ and $D$, which we use to set the invariants and guard conditions of the automaton $\mathcal{S}$. The set of admissible control actions—comprising \emph{accelerate} ($acc$), \emph{cruise} ($cru$), \emph{decelerate} ($dec$), \emph{brake} ($br$), and \emph{stop} ($st$)—is progressively controlled across these regions to ensure safety compliance.
%Within segment $A$, $acc$, $cru$, $dec$, and $br$ are enabled. Segment $B$ permits only $\{cru, dec, br\}$, while segment $C$ further restricts the system to $\{dec, br\}$. Finally, in segment $D$, $br$ and $st$ remains applicable to ensure safety compliance.
Given an output event stream (see Figure~\ref{fig:event-stream}) from the ACC controller, the enforcement function $\textit{E}_\varphi$ transforms it into a modified sequence (see Figure~\ref{fig:modified-events}, and Table~\ref{table:table1} in Appendix \ref{appendix}) that strictly satisfies the safety property $\varphi$ throughout the execution trace.

% \begin{figure}[htbp]
%     %\captionsetup{font=scriptsize}
%     \centering
%     \includegraphics[width=0.5\textwidth]{images/vel-dist-curve.drawio.png}
%     \caption{The speed-distance curve.}
%     \label{fig:vel-dist}
% \end{figure}

%In our enforcement algorithm, the function $\textit{Input}_\varphi$ computes this speed-distance curve from observed input specifications, and sets the parameters of the safety hybrid automaton $\mathcal{S}$ to initialize it to a standard hybrid automaton as shown in Figure~\ref{fig:monitor-ha}. For a given output event stream shown in Figure~\ref{fig:event-stream}, the $\textit{Edit}_\varphi$ produces a modified output event stream shown in Figure~\ref{fig:modified-events} that satisfies the property $\varphi$.

%Within our enforcement framework, the function $\textit{Input}_\varphi$ synthesizes the safe speed-distance curve from the observed environmental specifications. This function effectively instantiates the parametric hybrid automaton $\mathcal{S}$ by binding its parameters to concrete values to initialize it to a standard hybrid automaton, resulting in the operational monitor shown in Figure~\ref{fig:monitor-ha}.
%Given an output event stream (see Figure~\ref{fig:event-stream}) from the ACC controller, the enforcement function $\textit{Edit}_\varphi$ transforms it into a modified sequence (see Figure~\ref{fig:modified-events} and Table~\ref{table:table1}) that strictly satisfies the safety property $\varphi$ throughout the execution trace.

\paragraph{\textbf{Evaluation:}} The algorithm operates within a synchronous reactive framework, sampling events from the SuM at a uniform interval of $1$ second. For the execution trace illustrated in Figure~\ref{fig:event-stream}, the enforcer processed the entire sequence in $15.098$ s. Notably, the average response time per event was recorded at $4.67$ milliseconds. This low latency confirms that the framework introduces negligible computational overhead, ensuring that the enforcer can operate effectively within high-frequency reactive environments without violating the system's strict timing constraints. \textit{Implementation}: All experiments are performed on a machine with 8 GB RAM, Intel Core i5-8250U@1.60GHz, and 8-core processor with Ubuntu 18.04 64-bit OS. The problem files and the code base can be found at: \url{https://gitlab.com/anonymous2275/RuntimeEnforcer}.
%\sr{Depending on where we submit, the repository should be perhaps anonymous?}
\section{Related Works}

Runtime Enforcement (RE)~\cite{DBLP:journals/fmsd/FalconeMFR11,DBLP:journals/tissec/LigattiBW09,phdPinisetty15,DBLP:journals/fmsd/PinisettyFJMRN14,DBLP:journals/tissec/Schneider00} is an inherently online task that entails the active monitoring and mediation of a SuM to ensure that its execution adheres to a requirement policy.
%This process is managed by an \emph{Enforcer}, a component that intercepts system events and intervenes—according to the underlying requirement model—through actions such as suppressing, delaying, or inserting events.
%Consequently, RE is an inherently online task conducted during system execution.
When the specification includes temporal constraints, the discipline is referred to as real-time enforcement. This represents a more complex challenge than traditional runtime verification (RV)~\cite{DBLP:journals/tosem/BauerLS11,DBLP:conf/rv/Falcone10,DBLP:journals/sttt/FalconeFM12,DBLP:series/natosec/FalconeHR13,DBLP:conf/icse/PinisettySS18}, where the system is merely observed, and violations are reported rather than prevented.
An enforcement monitor (EM) transforms an (untrustworthy) input execution (sequence of events) into an output sequence of events that complies with a property (e.g., defining a desired safety requirement).
The applications of RE are extensive, spanning industrial safety protocols to regulatory compliance, and the field shares significant theoretical foundations with controller synthesis~\cite{DBLP:conf/icalp/AbadiLW89,DBLP:conf/popl/PnueliR89}.
Requirement policies are often divided into provisions and obligations~\cite{DBLP:conf/esorics/HiltyBP05}. Compliance with provisions is determined by the past and present behavior of the SuM; thus, the enforcer can function reactively based on the current system action. In contrast, compliance with obligations is contingent upon future system behavior. This necessitates a proactive enforcement strategy~\cite{DBLP:conf/csfw/BasinDH16}, where the enforcer must anticipate potential trajectories to preemptively prevent violations.

Several RE models have been proposed to address the increasing complexity of system dynamics, ranging from classical discrete-event frameworks to more recent architectures designed for real-time systems. While earlier models focused on untimed safety properties, recent advancements have introduced dense-time constraints and continuous-state evolutions.
Foundations for runtime enforcement were pioneered by Schneider in~\cite{DBLP:journals/tissec/Schneider00} and by Rinard in~\cite{DBLP:conf/oopsla/Rinard03}. 
\cite{DBLP:journals/tissec/Schneider00} discusses the class of safety policies that are enforceable via runtime monitoring. It provides an automata based mechanism for specifying policies that can be validated based on the prefixes of a system's execution. The enforcement logic focuses on safety properties, where the monitor intervenes by blocking the execution the moment it recognizes a sequence of actions that violates the target property. While~\cite{DBLP:conf/oopsla/Rinard03} identifies key properties that the execution must satisfy to be acceptable.It defines the system with a layered set of components, each of which enforces one of the acceptability properties.
%\cite{DBLP:journals/tissec/LigattiBW09} allows the enforcer to correct the input sequence by suppressing and/or inserting events, and the RE mechanisms proposed in~\cite{DBLP:journals/fmsd/FalconeMFR11} allow buffering events and releasing them upon observing a sequence that satisfies the desired property.
%Synthesis of enforcers for real-time properties expressed as dense TA has been studied~\cite{DBLP:journals/fmsd/PinisettyFJMRN14,DBLP:journals/scl/FalconeJMP16}.
%
%None of the above approaches are suitable for reactive systems since halting the program and delaying actions is not suitable. This is because for reactive systems the enforcer has to react instantaneously.
The framework introduced by Ligatti et al.~\cite{DBLP:journals/tissec/LigattiBW09} utilizes edit automata to correct input sequences through suppressing and/or inserting events. Similarly, the enforcement mechanisms proposed in~\cite{DBLP:journals/fmsd/FalconeMFR11} rely on buffering events and releasing them only upon the observation of a sequence that satisfies the desired policy. While the synthesis of enforcers for real-time properties—modeled via dense Timed Automata (TA)—has been extensively studied~\cite{DBLP:journals/fmsd/PinisettyFJMRN14,DBLP:journals/scl/FalconeJMP16}, these approaches are often inadequate for reactive systems. In such domains, halting execution or delaying actions are infeasible; rather, the enforcer must provide instantaneous rectification to maintain system liveness and satisfy strict timing constraints.
Thus, the RE approaches such as \cite{DBLP:journals/scl/FalconeJMP16,DBLP:conf/wodes/PinisettyFJM14,DBLP:conf/sac/PinisettyFJM14,DBLP:journals/fmsd/PinisettyFJMRN14} are not suitable for reactive systems.
%

%When considering reactive systems, terminating the system or delaying the reaction is not a feasible action for the enforcement monitor.
%Thus, the RE approaches such as \cite{wodesTE,Pinisetty2014,DBLP:journals/fmsd/PinisettyFJMRN14,DBLP:journals/scl/FalconeJMP16} are not suitable for reactive systems.

Shield synthesis~\cite{DBLP:journals/fmsd/KonighoferABHKT17} introduces an enforcement component, known as a \emph{safety shield}, to guarantee property adherence in reactive hardware systems at runtime. This approach mitigates the computational complexity associated with model checking and full reactive synthesis by focusing on a restricted subset of safety-critical properties, rather than the complete system specification. The shield operates by continuously monitoring the system's input/output interface and intervening to rectify erroneous outputs only when a safety violation is imminent.
The RE frameworks presented in \cite{DBLP:journals/tecs/PinisettyRSATH17,DBLP:conf/spin/PinisettyRSTH17} introduce bidirectional  enforcement monitoring for synchronous reactive systems. This approach provides a methodology for synthesizing EMs for properties specified via a variant of Discrete Timed Automata (DTA) \cite{DBLP:conf/charme/BozgaMT99,DBLP:journals/tcs/AlurD94}.
However, these frameworks consider discrete systems or monitoring at discrete time intervals leaving a gap in enforcement for complex systems that require proactive, continuous-time interventions to maintain safety.

%RE approaches in \cite{DBLP:journals/tecs/PinisettyRSATH17,DBLP:conf/spin/PinisettyRSTH17} present  enforcement frameworks suitable for synchronous reactive systems where the theory of bidirectional synchronous enforcement monitoring is introduced.
%A methodology for synthesis of EMs for properties defined using a variant of Discrete Timed Automata (DTA) \cite{DBLP:conf/charme/BozgaMT99,DBLP:journals/tcs/AlurD94} is proposed. 
%The framework considers similar constraints of soundness, and transparency as in the earlier mechanisms. However, the EM in the framework proposed in \cite{DBLP:journals/tecs/PinisettyRSATH17} also satisfies additional requirements such as instantaneity and causality which are particular to synchronous executions.

%Few works that consider monitoring hybrid systems:
%In~\cite{DBLP:conf/cav/CimattiMT11}, addresses the problem of finding traces of a network of hybrid automata that satisfying a property given as an Message Sequence Chart (MSC). The framework is finding a trace in the network of HAs that satisfy the event sequence of a given MSC.
%In~\cite{DBLP:conf/rv/SistlaZF11}, addresses the problem of monitoring a hybrid system, modeled as a Hidden Markov Chain (HMC), when the correctness specification is given by a deterministic Streett automaton.
%In~\cite{DBLP:conf/itsc/ChaiWLLH19}, presents a dynamic monitoring generation method for communications-based train control (CBTC) system using parametric hybrid automaton (PHA).
%However, these works presents methods for runtime monitoring and verifications of hybrid systems rather that runtime enforcement.

While several frameworks have been developed for hybrid system analysis, their scope is largely confined to monitoring and verification. Cimatti et al.~\cite{DBLP:conf/cav/CimattiMT11} address the problem of trace reconstruction within a network of hybrid automata to satisfy properties specified as Message Sequence Charts (MSCs). In the domain of stochastic models, Sistla et al.~\cite{DBLP:conf/rv/SistlaZF11} investigate the monitoring of hybrid systems modeled as Hidden Markov Chains (HMCs) against specifications defined by deterministic Streett automata. More recently, Chai et al.~\cite{DBLP:conf/itsc/ChaiWLLH19} proposed a dynamic monitoring generation method for communications-based train control (CBTC) systems utilizing Parametric Hybrid Automata (PHA). These studies focus primarily on runtime verification and passive monitoring—where the objective is to detect and report violations—rather than runtime enforcement.

In contrast, we model requirement specifications as hybrid automata, in the reactive setup our framework enables continuous monitoring of the system. Every clock tick events are observed and instantaneously rectified (suppression/delaying also handled with this). In the time period
between ticks, we also allow events to be inserted. Consequently, it can delay or insert events at arbitrary time instants, providing finer-grained control and ensuring property satisfaction in real time. As far as we know, no enforcement frameworks exist that consider expressive HA specification that are suitable for reactive systems.

\section{Conclusion and Future Work}

This paper presents a runtime enforcement framework based on Hybrid Automata monitoring specifications. The proposed approach extends existing runtime enforcement techniques by enabling continuous-time monitoring and proactive intervention for reactive cyber-physical systems. Unlike traditional approaches restricted to untimed or discrete-time specifications, our framework supports expressive hybrid-system properties and allows corrective actions through suppression, delay, and insertion of events both at synchronization points and between consecutive observation intervals.

We formally define the enforcement problem for safety hybrid automata, introduce enforcement constraints such as soundness, monotonicity, transparency, and completeness, and provide enforceability conditions for the proposed framework. An online enforcement algorithm is developed to synthesize safe execution traces while maintaining compliance with system timing constraints. The Adaptive Cruise Control (ACC) case study demonstrated how the enforcer can dynamically modify unsafe controller outputs and maintain critical safety requirements, such as minimum safe distance, throughout execution. Experimental evaluation shows that the framework operates with low latency and negligible runtime overhead, making it suitable for real-time reactive environments.

Overall, the work demonstrates that hybrid automata provide a powerful and expressive formalism for runtime enforcement of cyber-physical systems, enabling finer-grained control and stronger safety guarantees compared to existing discrete-time enforcement techniques.

\paragraph{\textbf{Future works:}} Several promising directions remain for future research. First, the current framework focuses primarily on safety properties represented as prefix-closed specifications. Extending the approach to support richer temporal properties, including liveness and obligation-based requirements, remains an important challenge. Second, the enforcement mechanism currently targets decidable subclasses of hybrid automata; future work may investigate scalable approximation techniques and symbolic reachability methods to support more expressive non-linear hybrid systems.
Another important direction is the integration of probabilistic and stochastic models to handle uncertainty arising from noisy sensors, imperfect environmental observations, and adversarial behaviors. Incorporating learning-based prediction mechanisms could further improve proactive enforcement by enabling the enforcer to anticipate unsafe trajectories earlier and synthesize more optimal corrective actions.
Future work will also explore distributed and decentralized enforcement architectures for multi-agent cyber-physical systems such as autonomous vehicle platoons, robotic swarms, and smart-grid infrastructures. Additionally, integrating shield synthesis techniques with hybrid automata-based enforcement may provide stronger correctness guarantees while reducing computational complexity.
Finally, implementing the framework on physical hardware platforms and evaluating it under real-world deployment conditions will be essential for validating its practical applicability, robustness, and scalability in safety-critical autonomous systems.

\clearpage
\bibliographystyle{plain}
\bibliography{mybib}

@phdthesis{phdPinisetty15,
	author    = {Srinivas Pinisetty},
	title     = {Runtime enforcement of timed properties. (Enforcement {\`{a}} l'{\'{e}}x{\'{e}}cution
	de propri{\'{e}}t{\'{e}}s temporis{\'{e}}es)},
	school    = {University of Rennes 1, France},
	year      = {2015},
	url       = {https://tel.archives-ouvertes.fr/tel-01185842},
	bibsource = {dblp computer science bibliography, https://dblp.org}
}

@inproceedings{DBLP:conf/icse/PinisettySS18,
  author       = {Srinivas Pinisetty and
                  Gerardo Schneider and
                  David Sands},
  editor       = {Stefania Gnesi and
                  Nico Plat and
                  Paola Spoletini and
                  Patrizio Pelliccione},
  title        = {Runtime verification of hyperproperties for deterministic programs},
  booktitle    = {Proceedings of the 6th Conference on Formal Methods in Software Engineering,
                  FormaliSE 2018, collocated with {ICSE} 2018, Gothenburg, Sweden, June
                  2, 2018},
  pages        = {20--29},
  publisher    = {{ACM}},
  year         = {2018},
  url          = {https://doi.org/10.1145/3193992.3193995},
  doi          = {10.1145/3193992.3193995},
  bibsource    = {dblp computer science bibliography, https://dblp.org}
}

@article{DBLP:journals/tosem/BauerLS11,
  author       = {Andreas Bauer and
                  Martin Leucker and
                  Christian Schallhart},
  title        = {Runtime Verification for {LTL} and {TLTL}},
  journal      = {{ACM} Trans. Softw. Eng. Methodol.},
  volume       = {20},
  number       = {4},
  pages        = {14:1--14:64},
  year         = {2011},
  url          = {https://doi.org/10.1145/2000799.2000800},
  doi          = {10.1145/2000799.2000800},
  bibsource    = {dblp computer science bibliography, https://dblp.org}
}

@article{ALUR19953,
title = "The algorithmic analysis of hybrid systems",
journal = "Theoretical Computer Science",
volume = "138",
number = "1",
pages = "3 - 34",
year = "1995",
note = "Hybrid Systems",
issn = "0304-3975",
doi = "https://doi.org/10.1016/0304-3975(94)00202-T",
url = "http://www.sciencedirect.com/science/article/pii/030439759400202T",
author = "R. Alur and C. Courcoubetis and N. Halbwachs and T.A. Henzinger and P.-H. Ho and X. Nicollin and A. Olivero and J. Sifakis and S. Yovine"
}

@inproceedings{DBLP:conf/hybrid/AlurCHH92,
  author       = {Rajeev Alur and
                  Costas Courcoubetis and
                  Thomas A. Henzinger and
                  Pei{-}Hsin Ho},
  editor       = {Robert L. Grossman and
                  Anil Nerode and
                  Anders P. Ravn and
                  Hans Rischel},
  title        = {Hybrid Automata: An Algorithmic Approach to the Specification and
                  Verification of Hybrid Systems},
  booktitle    = {Hybrid Systems},
  series       = {Lecture Notes in Computer Science},
  volume       = {736},
  pages        = {209--229},
  publisher    = {Springer},
  year         = {1992},
  url          = {https://doi.org/10.1007/3-540-57318-6\_30},
  doi          = {10.1007/3-540-57318-6\_30},
  bibsource    = {dblp computer science bibliography, https://dblp.org}
}

@inproceedings{DBLP:conf/focs/Pnueli77,
  author       = {Amir Pnueli},
  title        = {The Temporal Logic of Programs},
  booktitle    = {18th Annual Symposium on Foundations of Computer Science, Providence,
                  Rhode Island, USA, 31 October - 1 November 1977},
  pages        = {46--57},
  publisher    = {{IEEE} Computer Society},
  year         = {1977},
  url          = {https://doi.org/10.1109/SFCS.1977.32},
  doi          = {10.1109/SFCS.1977.32},
  bibsource    = {dblp computer science bibliography, https://dblp.org}
}

@inproceedings{DBLP:conf/formats/MalerN04,
  author       = {Oded Maler and
                  Dejan Nickovic},
  editor       = {Yassine Lakhnech and
                  Sergio Yovine},
  title        = {Monitoring Temporal Properties of Continuous Signals},
  booktitle    = {Formal Techniques, Modelling and Analysis of Timed and Fault-Tolerant
                  Systems, Joint International Conferences on Formal Modelling and Analysis
                  of Timed Systems, {FORMATS} 2004 and Formal Techniques in Real-Time
                  and Fault-Tolerant Systems, {FTRTFT} 2004, Grenoble, France, September
                  22-24, 2004, Proceedings},
  series       = {Lecture Notes in Computer Science},
  volume       = {3253},
  pages        = {152--166},
  publisher    = {Springer},
  year         = {2004},
  url          = {https://doi.org/10.1007/978-3-540-30206-3\_12},
  doi          = {10.1007/978-3-540-30206-3\_12},
  bibsource    = {dblp computer science bibliography, https://dblp.org}
}

@article{DBLP:journals/scl/FalconeJMP16,
  author       = {Yli{\`{e}}s Falcone and
                  Thierry J{\'{e}}ron and
                  Herv{\'{e}} Marchand and
                  Srinivas Pinisetty},
  title        = {Runtime enforcement of regular timed properties by suppressing and
                  delaying events},
  journal      = {Sci. Comput. Program.},
  volume       = {123},
  pages        = {2--41},
  year         = {2016},
  url          = {https://doi.org/10.1016/j.scico.2016.02.008},
  doi          = {10.1016/J.SCICO.2016.02.008},
  bibsource    = {dblp computer science bibliography, https://dblp.org}
}

@article{DBLP:journals/sttt/FalconeFM12,
  author       = {Yli{\`{e}}s Falcone and
                  Jean{-}Claude Fernandez and
                  Laurent Mounier},
  title        = {What can you verify and enforce at runtime?},
  journal      = {Int. J. Softw. Tools Technol. Transf.},
  volume       = {14},
  number       = {3},
  pages        = {349--382},
  year         = {2012},
  url          = {https://doi.org/10.1007/s10009-011-0196-8},
  doi          = {10.1007/S10009-011-0196-8},
  bibsource    = {dblp computer science bibliography, https://dblp.org}
}

@incollection{DBLP:series/natosec/FalconeHR13,
  author       = {Yli{\`{e}}s Falcone and
                  Klaus Havelund and
                  Giles Reger},
  editor       = {Manfred Broy and
                  Doron A. Peled and
                  Georg Kalus},
  title        = {A Tutorial on Runtime Verification},
  booktitle    = {Engineering Dependable Software Systems},
  series       = {{NATO} Science for Peace and Security Series, {D:} Information and
                  Communication Security},
  volume       = {34},
  pages        = {141--175},
  publisher    = {{IOS} Press},
  year         = {2013},
  url          = {https://doi.org/10.3233/978-1-61499-207-3-141},
  doi          = {10.3233/978-1-61499-207-3-141},
  bibsource    = {dblp computer science bibliography, https://dblp.org}
}

@article{DBLP:journals/fmsd/PinisettyFJMRN14,
  author       = {Srinivas Pinisetty and
                  Yli{\`{e}}s Falcone and
                  Thierry J{\'{e}}ron and
                  Herv{\'{e}} Marchand and
                  Antoine Rollet and
                  Omer Nguena{-}Timo},
  title        = {Runtime enforcement of timed properties revisited},
  journal      = {Formal Methods Syst. Des.},
  volume       = {45},
  number       = {3},
  pages        = {381--422},
  year         = {2014},
  url          = {https://doi.org/10.1007/s10703-014-0215-y},
  doi          = {10.1007/S10703-014-0215-Y},
  bibsource    = {dblp computer science bibliography, https://dblp.org}
}

@article{DBLP:journals/tissec/LigattiBW09,
  author       = {Jay Ligatti and
                  Lujo Bauer and
                  David Walker},
  title        = {Run-Time Enforcement of Nonsafety Policies},
  journal      = {{ACM} Trans. Inf. Syst. Secur.},
  volume       = {12},
  number       = {3},
  pages        = {19:1--19:41},
  year         = {2009},
  url          = {https://doi.org/10.1145/1455526.1455532},
  doi          = {10.1145/1455526.1455532},
  bibsource    = {dblp computer science bibliography, https://dblp.org}
}

@article{DBLP:journals/tissec/Schneider00,
  author       = {Fred B. Schneider},
  title        = {Enforceable security policies},
  journal      = {{ACM} Trans. Inf. Syst. Secur.},
  volume       = {3},
  number       = {1},
  pages        = {30--50},
  year         = {2000},
  url          = {https://doi.org/10.1145/353323.353382},
  doi          = {10.1145/353323.353382},
  bibsource    = {dblp computer science bibliography, https://dblp.org}
}

@article{DBLP:journals/fmsd/FalconeMFR11,
  author       = {Yli{\`{e}}s Falcone and
                  Laurent Mounier and
                  Jean{-}Claude Fernandez and
                  Jean{-}Luc Richier},
  title        = {Runtime enforcement monitors: composition, synthesis, and enforcement
                  abilities},
  journal      = {Formal Methods Syst. Des.},
  volume       = {38},
  number       = {3},
  pages        = {223--262},
  year         = {2011},
  url          = {https://doi.org/10.1007/s10703-011-0114-4},
  doi          = {10.1007/S10703-011-0114-4},
  bibsource    = {dblp computer science bibliography, https://dblp.org}
}

@inproceedings{DBLP:conf/icalp/AbadiLW89,
  author       = {Mart{\'{\i}}n Abadi and
                  Leslie Lamport and
                  Pierre Wolper},
  editor       = {Giorgio Ausiello and
                  Mariangiola Dezani{-}Ciancaglini and
                  Simona Ronchi Della Rocca},
  title        = {Realizable and Unrealizable Specifications of Reactive Systems},
  booktitle    = {Automata, Languages and Programming, 16th International Colloquium,
                  ICALP89, Stresa, Italy, July 11-15, 1989, Proceedings},
  series       = {Lecture Notes in Computer Science},
  volume       = {372},
  pages        = {1--17},
  publisher    = {Springer},
  year         = {1989},
  url          = {https://doi.org/10.1007/BFb0035748},
  doi          = {10.1007/BFB0035748},
  bibsource    = {dblp computer science bibliography, https://dblp.org}
}

@inproceedings{DBLP:conf/popl/PnueliR89,
  author       = {Amir Pnueli and
                  Roni Rosner},
  title        = {On the Synthesis of a Reactive Module},
  booktitle    = {Conference Record of the Sixteenth Annual {ACM} Symposium on Principles
                  of Programming Languages, Austin, Texas, USA, January 11-13, 1989},
  pages        = {179--190},
  publisher    = {{ACM} Press},
  year         = {1989},
  url          = {https://doi.org/10.1145/75277.75293},
  doi          = {10.1145/75277.75293},
  bibsource    = {dblp computer science bibliography, https://dblp.org}
}

@inproceedings{DBLP:conf/esorics/HiltyBP05,
  author       = {Manuel Hilty and
                  David A. Basin and
                  Alexander Pretschner},
  editor       = {Sabrina De Capitani di Vimercati and
                  Paul F. Syverson and
                  Dieter Gollmann},
  title        = {On Obligations},
  booktitle    = {Computer Security - {ESORICS} 2005, 10th European Symposium on Research
                  in Computer Security, Milan, Italy, September 12-14, 2005, Proceedings},
  series       = {Lecture Notes in Computer Science},
  volume       = {3679},
  pages        = {98--117},
  publisher    = {Springer},
  year         = {2005},
  url          = {https://doi.org/10.1007/11555827\_7},
  doi          = {10.1007/11555827\_7},
  bibsource    = {dblp computer science bibliography, https://dblp.org}
}

@inproceedings{DBLP:conf/csfw/BasinDH16,
  author       = {David A. Basin and
                  S{\o}ren Debois and
                  Thomas T. Hildebrandt},
  title        = {In the Nick of Time: Proactive Prevention of Obligation Violations},
  booktitle    = {{IEEE} 29th Computer Security Foundations Symposium, {CSF} 2016, Lisbon,
                  Portugal, June 27 - July 1, 2016},
  pages        = {120--134},
  publisher    = {{IEEE} Computer Society},
  year         = {2016},
  url          = {https://doi.org/10.1109/CSF.2016.16},
  doi          = {10.1109/CSF.2016.16},
  bibsource    = {dblp computer science bibliography, https://dblp.org}
}

@inproceedings{DBLP:conf/itsc/ChaiWLLH19,
  author       = {Ming Chai and
                  Haifeng Wang and
                  Hongjie Liu and
                  Jidong Lv and
                  Qian Hu},
  title        = {Runtime Verification of Communications-based Train Control with Parametric
                  Hybrid Automata},
  booktitle    = {2019 {IEEE} Intelligent Transportation Systems Conference, {ITSC}
                  2019, Auckland, New Zealand, October 27-30, 2019},
  pages        = {2160--2165},
  publisher    = {{IEEE}},
  year         = {2019},
  url          = {https://doi.org/10.1109/ITSC.2019.8917282},
  doi          = {10.1109/ITSC.2019.8917282},
  bibsource    = {dblp computer science bibliography, https://dblp.org}
}

@inproceedings{DBLP:conf/lics/Henzinger96,
  author       = {Thomas A. Henzinger},
  title        = {The Theory of Hybrid Automata},
  booktitle    = {Proceedings, 11th Annual {IEEE} Symposium on Logic in Computer Science,
                  New Brunswick, New Jersey, USA, July 27-30, 1996},
  pages        = {278--292},
  publisher    = {{IEEE} Computer Society},
  year         = {1996},
  url          = {https://doi.org/10.1109/LICS.1996.561342},
  doi          = {10.1109/LICS.1996.561342},
  bibsource    = {dblp computer science bibliography, https://dblp.org}
}

@article{DBLP:journals/jcs/BasinDH24,
  author       = {David A. Basin and
                  S{\o}ren Debois and
                  Thomas T. Hildebrandt},
  title        = {Proactive enforcement of provisions and obligations},
  journal      = {J. Comput. Secur.},
  volume       = {32},
  number       = {3},
  pages        = {247--289},
  year         = {2024},
  url          = {https://doi.org/10.3233/JCS-210078},
  doi          = {10.3233/JCS-210078},
  bibsource    = {dblp computer science bibliography, https://dblp.org}
}

@inproceedings{DBLP:conf/stoc/HenzingerKPV95,
  author       = {Thomas A. Henzinger and
                  Peter W. Kopke and
                  Anuj Puri and
                  Pravin Varaiya},
  editor       = {Frank Thomson Leighton and
                  Allan Borodin},
  title        = {What's decidable about hybrid automata?},
  booktitle    = {Proceedings of the Twenty-Seventh Annual {ACM} Symposium on Theory
                  of Computing, 29 May-1 June 1995, Las Vegas, Nevada, {USA}},
  pages        = {373--382},
  publisher    = {{ACM}},
  year         = {1995},
  url          = {https://doi.org/10.1145/225058.225162},
  doi          = {10.1145/225058.225162},
  bibsource    = {dblp computer science bibliography, https://dblp.org}
}

@inproceedings{DBLP:conf/rv/Falcone10,
  author       = {Yli{\`{e}}s Falcone},
  editor       = {Howard Barringer and
                  Yli{\`{e}}s Falcone and
                  Bernd Finkbeiner and
                  Klaus Havelund and
                  Insup Lee and
                  Gordon J. Pace and
                  Grigore Rosu and
                  Oleg Sokolsky and
                  Nikolai Tillmann},
  title        = {You Should Better Enforce Than Verify},
  booktitle    = {Runtime Verification - First International Conference, {RV} 2010,
                  St. Julians, Malta, November 1-4, 2010. Proceedings},
  series       = {Lecture Notes in Computer Science},
  pages        = {89--105},
  publisher    = {Springer},
  year         = {2010},
  url          = {https://doi.org/10.1007/978-3-642-16612-9\_9},
  doi          = {10.1007/978-3-642-16612-9\_9},
  bibsource    = {dblp computer science bibliography, https://dblp.org}
}

@article{DBLP:journals/tecs/PinisettyRSATH17,
  author       = {Srinivas Pinisetty and
                  Partha S. Roop and
                  Steven Smyth and
                  Nathan Allen and
                  Stavros Tripakis and
                  Reinhard von Hanxleden},
  title        = {Runtime Enforcement of Cyber-Physical Systems},
  journal      = {{ACM} Trans. Embed. Comput. Syst.},
  volume       = {16},
  number       = {5s},
  pages        = {178:1--178:25},
  year         = {2017},
  url          = {https://doi.org/10.1145/3126500},
  doi          = {10.1145/3126500},
  bibsource    = {dblp computer science bibliography, https://dblp.org}
}

@inproceedings{DBLP:conf/asm/HoudekR20,
  author       = {Frank Houdek and
                  Alexander Raschke},
  editor       = {Alexander Raschke and
                  Dominique M{\'{e}}ry and
                  Frank Houdek},
  title        = {Adaptive Exterior Light and Speed Control System},
  booktitle    = {Rigorous State-Based Methods - 7th International Conference, {ABZ}
                  2020, Ulm, Germany, May 27-29, 2020, Proceedings},
  series       = {Lecture Notes in Computer Science},
  pages        = {281--301},
  publisher    = {Springer},
  year         = {2020},
  url          = {https://doi.org/10.1007/978-3-030-48077-6\_24},
  doi          = {10.1007/978-3-030-48077-6\_24},
  bibsource    = {dblp computer science bibliography, https://dblp.org}
}

@inproceedings{DBLP:conf/tacas/BogomolovFGH17,
  author       = {Sergiy Bogomolov and
                  Goran Frehse and
                  Mirco Giacobbe and
                  Thomas A. Henzinger},
  editor       = {Axel Legay and
                  Tiziana Margaria},
  title        = {Counterexample-Guided Refinement of Template Polyhedra},
  booktitle    = {Tools and Algorithms for the Construction and Analysis of Systems
                  - 23rd International Conference, {TACAS} 2017, Held as Part of the
                  European Joint Conferences on Theory and Practice of Software, {ETAPS}
                  2017, Uppsala, Sweden, April 22-29, 2017, Proceedings, Part {I}},
  series       = {Lecture Notes in Computer Science},
  pages        = {589--606},
  year         = {2017},
  url          = {https://doi.org/10.1007/978-3-662-54577-5\_34},
  doi          = {10.1007/978-3-662-54577-5\_34},
  bibsource    = {dblp computer science bibliography, https://dblp.org}
}

@inproceedings{DBLP:conf/kbse/LegunsenHXRM16,
  author       = {Owolabi Legunsen and
                  Wajih Ul Hassan and
                  Xinyue Xu and
                  Grigore Rosu and
                  Darko Marinov},
  editor       = {David Lo and
                  Sven Apel and
                  Sarfraz Khurshid},
  title        = {How good are the specs? a study of the bug-finding effectiveness of
                  existing Java {API} specifications},
  booktitle    = {Proceedings of the 31st {IEEE/ACM} International Conference on Automated
                  Software Engineering, {ASE} 2016, Singapore, September 3-7, 2016},
  pages        = {602--613},
  publisher    = {{ACM}},
  year         = {2016},
  url          = {https://doi.org/10.1145/2970276.2970356},
  doi          = {10.1145/2970276.2970356},
  bibsource    = {dblp computer science bibliography, https://dblp.org}
}

@inproceedings{DBLP:conf/charme/BozgaMT99,
  author       = {Marius Bozga and
                  Oded Maler and
                  Stavros Tripakis},
  editor       = {Laurence Pierre and
                  Thomas Kropf},
  title        = {Efficient Verification of Timed Automata Using Dense and Discrete
                  Time Semantics},
  booktitle    = {Correct Hardware Design and Verification Methods, 10th {IFIP} {WG}
                  10.5 Advanced Research Working Conference, {CHARME} '99, Bad Herrenalb,
                  Germany, September 27-29, 1999, Proceedings},
  series       = {Lecture Notes in Computer Science},
  pages        = {125--141},
  publisher    = {Springer},
  year         = {1999},
  url          = {https://doi.org/10.1007/3-540-48153-2\_11},
  doi          = {10.1007/3-540-48153-2\_11},
  bibsource    = {dblp computer science bibliography, https://dblp.org}
}

@inproceedings{DBLP:conf/spin/PinisettyRSTH17,
  author       = {Srinivas Pinisetty and
                  Partha S. Roop and
                  Steven Smyth and
                  Stavros Tripakis and
                  Reinhard von Hanxleden},
  editor       = {Hakan Erdogmus and
                  Klaus Havelund},
  title        = {Runtime enforcement of reactive systems using synchronous enforcers},
  booktitle    = {Proceedings of the 24th {ACM} {SIGSOFT} International {SPIN} Symposium
                  on Model Checking of Software, Santa Barbara, CA, USA, July 10-14,
                  2017},
  pages        = {80--89},
  publisher    = {{ACM}},
  year         = {2017},
  url          = {https://doi.org/10.1145/3092282.3092291},
  doi          = {10.1145/3092282.3092291},
  bibsource    = {dblp computer science bibliography, https://dblp.org}
}

@inproceedings{DBLP:conf/oopsla/Rinard03,
  author       = {Martin C. Rinard},
  editor       = {Ron Crocker and
                  Guy L. Steele Jr.},
  title        = {Acceptability-oriented computing},
  booktitle    = {Companion of the 18th Annual {ACM} {SIGPLAN} Conference on Object-Oriented
                  Programming, Systems, Languages, and Applications, {OOPSLA} 2003,
                  October 26-30, 2003, Anaheim, CA, {USA}},
  pages        = {221--239},
  publisher    = {{ACM}},
  year         = {2003},
  url          = {https://doi.org/10.1145/949344.949402},
  doi          = {10.1145/949344.949402},
  bibsource    = {dblp computer science bibliography, https://dblp.org}
}

@inproceedings{DBLP:conf/wodes/PinisettyFJM14,
  author       = {Srinivas Pinisetty and
                  Yli{\`{e}}s Falcone and
                  Thierry J{\'{e}}ron and
                  Herv{\'{e}} Marchand},
  editor       = {Jean{-}Jacques Lesage and
                  Jean{-}Marc Faure and
                  Jos{\'{e}} E. R. Cury and
                  Bengt Lennartson},
  title        = {Runtime Enforcement of Parametric Timed Properties with Practical
                  Applications},
  booktitle    = {12th International Workshop on Discrete Event Systems, {WODES} 2014,
                  Cachan, France, May 14-16, 2014},
  pages        = {420--427},
  publisher    = {International Federation of Automatic Control},
  year         = {2014},
  url          = {https://doi.org/10.3182/20140514-3-FR-4046.00041},
  doi          = {10.3182/20140514-3-FR-4046.00041},
  bibsource    = {dblp computer science bibliography, https://dblp.org}
}

@inproceedings{DBLP:conf/sac/PinisettyFJM14,
  author       = {Srinivas Pinisetty and
                  Yli{\`{e}}s Falcone and
                  Thierry J{\'{e}}ron and
                  Herv{\'{e}} Marchand},
  editor       = {Yookun Cho and
                  Sung Y. Shin and
                  Sang{-}Wook Kim and
                  Chih{-}Cheng Hung and
                  Jiman Hong},
  title        = {Runtime enforcement of regular timed properties},
  booktitle    = {Symposium on Applied Computing, {SAC} 2014, Gyeongju, Republic of
                  Korea - March 24 - 28, 2014},
  pages        = {1279--1286},
  publisher    = {{ACM}},
  year         = {2014},
  url          = {https://doi.org/10.1145/2554850.2554967},
  doi          = {10.1145/2554850.2554967},
  bibsource    = {dblp computer science bibliography, https://dblp.org}
}

@article{DBLP:journals/fmsd/KonighoferABHKT17,
  author       = {Bettina K{\"{o}}nighofer and
                  Mohammed Alshiekh and
                  Roderick Bloem and
                  Laura R. Humphrey and
                  Robert K{\"{o}}nighofer and
                  Ufuk Topcu and
                  Chao Wang},
  title        = {Shield synthesis},
  journal      = {Formal Methods Syst. Des.},
  volume       = {51},
  number       = {2},
  pages        = {332--361},
  year         = {2017},
  url          = {https://doi.org/10.1007/s10703-017-0276-9},
  doi          = {10.1007/S10703-017-0276-9},
  bibsource    = {dblp computer science bibliography, https://dblp.org}
}

@article{DBLP:journals/tcs/AlurD94,
  author       = {Rajeev Alur and
                  David L. Dill},
  title        = {A Theory of Timed Automata},
  journal      = {Theor. Comput. Sci.},
  volume       = {126},
  number       = {2},
  pages        = {183--235},
  year         = {1994},
  url          = {https://doi.org/10.1016/0304-3975(94)90010-8},
  doi          = {10.1016/0304-3975(94)90010-8},
  bibsource    = {dblp computer science bibliography, https://dblp.org}
}

@inproceedings{DBLP:conf/cav/CimattiMT11,
  author       = {Alessandro Cimatti and
                  Sergio Mover and
                  Stefano Tonetta},
  editor       = {Ganesh Gopalakrishnan and
                  Shaz Qadeer},
  title        = {Efficient Scenario Verification for Hybrid Automata},
  booktitle    = {Computer Aided Verification - 23rd International Conference, {CAV}
                  2011, Snowbird, UT, USA, July 14-20, 2011. Proceedings},
  series       = {Lecture Notes in Computer Science},
  pages        = {317--332},
  publisher    = {Springer},
  year         = {2011},
  url          = {https://doi.org/10.1007/978-3-642-22110-1\_25},
  doi          = {10.1007/978-3-642-22110-1\_25},
  bibsource    = {dblp computer science bibliography, https://dblp.org}
}

@inproceedings{DBLP:conf/rv/SistlaZF11,
  author       = {A. Prasad Sistla and
                  Milos Zefran and
                  Yao Feng},
  editor       = {Sarfraz Khurshid and
                  Koushik Sen},
  title        = {Runtime Monitoring of Stochastic Cyber-Physical Systems with Hybrid
                  State},
  booktitle    = {Runtime Verification - Second International Conference, {RV} 2011,
                  San Francisco, CA, USA, September 27-30, 2011, Revised Selected Papers},
  series       = {Lecture Notes in Computer Science},
  pages        = {276--293},
  publisher    = {Springer},
  year         = {2011},
  url          = {https://doi.org/10.1007/978-3-642-29860-8\_21},
  doi          = {10.1007/978-3-642-29860-8\_21},
  bibsource    = {dblp computer science bibliography, https://dblp.org}
}
\appendix
\section{Appendix- Illustration of execution of the enforcement algorithm}\label{appendix}

\begin{table}[htbp]
  \centering\scriptsize
    \begin{tabular}{|c|c|c|c|c|c|c|}
    \hline
        %\multicolumn{2}{|c|}{\multirow{2}{*}{Benchmarks}} & \multirow{2}{*}{Depth} & \multicolumn{3}{c|}{Time (in secs)}  & \multirow{2}{*}{AT}\\\cline{4-6}
        \multirow{2}{*}{$t$} & {Output} & {Current} & {Current values} & {Enforced} & \multirow{2}{*}{Next Mode} & \multirow{2}{*}{Remark}\\
        & Event & Mode & $(d,v,a)$ & Event & & \\
        %& Event & State & Event & State \\
    \hline
       $0.0$ & $stutter$ & \textit{InitMode} & $95,10,0$ & $stutter$ & \textit{InitMode} & System dwelling\\
       $1.0$ & \textit{acc} & \textit{InitMode} & $98.89,10,0$  & $stutter$ & \textit{InitMode} & Suppress/delay\\
       $1.29$ & $-$ & \textit{InitMode} & $100.02,10,1$  & $acc$ & \textit{AccelerateMode} & Event delayed\\
       $2.0$ & $acc$ & \textit{AccelerateMode} & $100.56,10.72,2$  & $acc$ & \textit{AccelerateMode} & Event executed\\
       $3.0$ & $acc$ & \textit{AccelerateMode} & $104.74,12.72,3$  & $acc$ & \textit{AccelerateMode} & Event executed\\
       $4.0$ & $stutter$ & \textit{AccelerateMode} & $104.42,15.72,3$  & $stutter$ & \textit{AccelerateMode} & System dwelling\\
       $5.0$ & $stutter$ & \textit{AccelerateMode} & $101.11,18.72,3$  & $stutter$ & \textit{AccelerateMode} & System dwelling\\
       $5.22$ & $-$ & \textit{AccelerateMode} & $99.98,19.38,0$  & $cru$ & \textit{CruiseMode} & Event inserted\\
       $6.0$ & $cru$ & \textit{CruiseMode} & $95.64,19.38,0$  & $cru$ & \textit{CruiseMode} & Event executed\\
       $7.0$ & $stutter$ & \textit{CruiseMode} & $9.15,19.38,0$  & $stutter$ & \textit{CruiseMode} & System dwelling\\
       $8.0$ & $stutter$ & \textit{CruiseMode} & $84.82,19.38,0$  & $stutter$ & \textit{CruiseMode} & System dwelling\\
       $9.0$ & $acc$ & \textit{CruiseMode} & $79.17,19.38,0$  & $stutter$ & \textit{CruiseMode} & Suppressed\\
       $10.0$ & $stutter$ & \textit{CruiseMode} & $73.84,19.38,0$  & $stutter$ & \textit{CruiseMode} & System dwelling\\
       $10.67$ & $-$ & \textit{CruiseMode} & $69.99,19.38,-1$  & $dec$ & \textit{DecelerateMode} & Event inserted\\
       $11.0$ & $stutter$ & \textit{DecelerateMode} & $68.18,19.04,-1$  & $stutter$ & \textit{DecelerateMode} & System dwelling\\
       $12.0$ & $stutter$ & \textit{DecelerateMode} & $63.53,18.04,-1$  & $stutter$ & \textit{DecelerateMode} & System dwelling\\
       $13.0$ & $stutter$ & \textit{DecelerateMode} & $59.87,17.04,-1$  & $stutter$ & \textit{DecelerateMode} & System dwelling\\
       $14.0$ & $stutter$ & \textit{DecelerateMode} & $57.22,16.04,-1$  & $stutter$ & \textit{DecelerateMode} & System dwelling\\
       $15.0$ & $stutter$ & \textit{DecelerateMode} & $55.56,15.04,-1$  & $stutter$ & \textit{DecelerateMode} & System dwelling\\
       
    \hline
    \end{tabular}
  \vspace{10pt}
  \caption{Presents the enforced events that satisfy the safety property $\varphi$ for the event sequence from the system shown in Figure~\ref{fig:event-stream}. N.B. $stutter$ represents a stutter event.}
  \label{table:table1}
%  \vspace{-20pt}
\end{table}

A detailed execution of the enforcement algorithm (Algorithm~\ref{algo:enforcer}) is presented in Table~\ref{table:table1}. $t$ and \emph{Output Event} represent the physical time, and at which an event comes from the system, respectively. \emph{Current Mode} and \emph{Current Values} signify the current location where control resides and current state of the automaton. \emph{Enforced Event} presents the rectified events from the enforcer. \emph{Next Mode} shows the transition of the control in the automaton after the enforced event. \emph{Remark} provides comments on the execution.

\section{Appendix- Proof of Proposition~\ref{prop:prop2}} \label{appendix2}
\begin{proof} [Correctness of the enforcement algorithm]
The correctness of the algorithm is established via induction over the discrete execution steps of the enforced output sequence $\sigma'$. Let $E^*_{\varphi}$ denote the enforcement function. Inductive Hypothesis: Assume that up to time $t$, the enforcer has processed an input stream $\sigma$ to produce a sound output sequence $\sigma' = E^*_{\varphi}(\sigma)$ such that $\sigma' \models \varphi$.We analyze the inductive step during the transition to time $t+1$ under three exhaustive cases:
\begin{itemize}
    \item Case 1 (No modification): A new event $e$ is observed at time $t$. If the extension of the trace satisfies the property, i.e., $\sigma' \cdot e \models \varphi$, Algorithm~\ref{algo:enforcer} outputs the event without modification. Thus, $E^*_{\varphi}(\sigma \cdot e) = \sigma' \cdot e \models \varphi$.
    \item Case 2 (Suppression/Substitution): If the incoming event causes a violation, i.e., $\sigma' \cdot e \not\models \varphi$, the enforcer intercepts and suppresses $e$, substituting it with a safe stutter event $d$. This ensures that $E^*_{\varphi}(\sigma \cdot e) = \sigma' \cdot d \models \varphi$.
    \item Case 3 (Proactive Inter-tick Insertion): If no boundary event causes a violation, but continuous state evolution within the dense-time interval $\tau \in (t, t+1)$ threatens a state invariant violation, the proactive look-ahead mechanism is triggered. At the critical moment $\tau$, Algorithm~\ref{algo:enforcer} inserts a corrective event $e'$ to steer the trajectory back into the safe space, ensuring $\sigma' \cdot e' \models \varphi$.
\end{itemize}
Consequently, for all valid transitions, the output trace at $t+1$ preserves safety ($E^*_{\varphi} \models \varphi$), completing the proof. $\Box$
\end{proof}

\end{document}